\documentclass[envcountsame]{llncs}

\makeatletter
\renewcommand*\l@author[2]{}
\renewcommand*\l@title[2]{}
\makeatletter

\scrollmode
\usepackage{amsmath}
\usepackage{amsfonts}
\usepackage{amssymb}
\usepackage{latexsym}
\usepackage{stmaryrd}
\usepackage{array}
\usepackage{exscale}
\newcommand{\nc}{\newcommand}
\newcommand{\ol}{\overline}

\newcommand{\es}{\emptyset}
\newcommand{\sm}{\setminus}
\newcommand{\ve}{\varepsilon}

\newcommand{\bca}{\bigcap}
\newcommand{\Lra}{\Leftrightarrow}

\newcommand{\Ra}{\Rightarrow}

\newcommand{\ra}{\rightarrow}

\newcommand{\lra}{\leftrightarrow}

\newcommand{\sse}{\subseteq}

\newcommand{\spe}{\supseteq}
\newcommand{\fa}{\forall}

\newcommand{\mc}{\mathcal}

\newcommand{\DMO}{\DeclareMathOperator}
\newcommand{\DST}{\displaystyle}

\newcommand{\NN}{\mathbb{N}}
\newcommand{\NNZ}{\NN_0}

\newcommand{\RR}{\mathbb{R}}

\newcommand{\und}{{\:\wedge\:}} % and
\newcommand{\mb}{{\:|\:}} % Mengenbildner; set creation
\newcommand{\set}[1]{\{ #1 \}}
\newcommand{\setb}[1]{\big \{ \, #1 \, \big \}}
\nc{\simlvi}[1]{\!\sim_{#1}}
\nc{\apprel}[3]{{#1}(#2)_{(#3)}} % \apprel{R}{M}{k} : Anwendung der Relation auf M bzgl. Stufe k; application of relation R to set M at level k
\newcommand{\tb}[2]{\set{#1, \dots, #2}} % Teilbereich
\providecommand{\abs}[1]{\lvert #1 \rvert} % 29.8.2000
\nc{\sselr}{\sse^{\mapsto}}
\nc{\sserl}{\sse^{\mapsfrom}}
\nc{\spelr}{\spe^{\mapsto}}
\nc{\sperl}{\spe^{\mapsfrom}}
\newcommand{\Va}{\mc{V\hspace{-0.1em}A}}

\newcommand{\Lit}{\mc{LIT}}

\newcommand{\Cls}{\mc{CLS}}

\newcommand{\Sat}{\mc{SAT}}

\newcommand{\Usat}{\mc{USAT}}

\newcommand{\Musat}{\mc{M\hspace{0.8pt}U}} % minimally unsatisfiable clause-sets; minimal unerfuellbare Klauselmengen
\newcommand{\Musati}[1]{\Musat_{\!#1}} % used to select subsets, e.g., $\Musati{\delta=k}$; zur Auswahl von Teilklassen
\newcommand{\Smusat}{\mc{S}\Musat} % saturated (``maximal'') minimally unsatisfiable clause-sets; saturierte (``maximale'') minimal unerfuellbare Klauselmengen
\newcommand{\Smusati}[1]{\Smusat_{\!#1}}
\nc{\Clsoo}{\Cls^{1,1}} % jede Variable hat 1-1-Vorkommen; each variable occurs once positively and once negatively

\DeclareMathOperator{\var}{var}

\newcommand{\Clash}{\mc{HIT}} % schwach-resolvierbare Klauselmengen (global umdefiniert 24.7.2003)

\newcommand{\Uclash}{\mc{U}\Clash} % unerfuellbare Treffklauselmengen
\newcommand{\Uclashi}[1]{\Uclash_{\!\!#1}}
\newcommand{\Ho}{\mc{HO}} % Hornformeln; Horn clause-sets
\newcommand{\Rho}{\mc{R}\Ho} % Umbenennbare Hornformeln; renamable Horn clause-sets

\DeclareMathOperator{\res}{\diamond} % die (partielle) Resolutionsoperation
\DeclareMathOperator{\dpl}{DP} % der DP-Operator
\newcommand{\dpi}[1]{\dpl_{\!#1}}

\newcommand{\pao}[2]{\langle #1 \ra #2 \rangle}

\DMO{\rsub}{r_S} % subsumption-elimination
\DMO{\rk}{r} % propagation-reduction
\DMO{\rki}{r_{\infty}} % the limit of r_k
\DeclareMathOperator{\ldeg}{ld} % Literalgrad; literal degree
\DeclareMathOperator{\vdeg}{vd} % Variablengrad; variable degree
\DMO{\varsing}{\var_s} % singular variables
\DMO{\varosing}{\var_{1s}} % 1-singular variables
\DMO{\varnosing}{\var_{\neg1s}} % non-1-singular variables
\nc{\Musatns}{\Musat'} % non-singular minimally unsatisfiable clause-sets
\nc{\Musatnsi}[1]{\Musati{#1}'}
\nc{\Smusatns}{\Smusat'} % non-singular saturated minimally unsatisfiable clause-sets
\nc{\Smusatnsi}[1]{\Smusati{#1}'}
\nc{\Uclashns}{\Uclash'} % non-singular hitting clause-sets
\nc{\Uclashnsi}[1]{\Uclashi{#1}'}
\nc{\tsdp}{\xrightarrow{\text{sDP}}}
\nc{\tsdps}{\tsdp_{\!*}}
\nc{\tosdp}{\xrightarrow{\text{1sDP}}}
\nc{\tosdps}{\tosdp_{\!*}}
\DMO{\sdp}{sDP} % the set of reachable non-singular clause-sets
\DMO{\osdp}{sDP_1} % the normalform obtained by 1sDP-reduction
\nc{\cflmusat}{\mc{CF}\Musat} % F in MU with confluent singular DP-reduction
\nc{\cflmusati}[1]{\mc{CF}\Musati{#1}}
\nc{\cflimusat}{\mc{CFI}\Musat} % F in MU with confluent singular DP-reduction
\DMO{\sNF}{sNF} % singular normal form
\DMO{\eqp}{eqp} % equality-preserving permutations
\DMO{\sgp}{sp} % singularity-preserving permutations
\DMO{\singind}{si} % singularity index
\DMO{\osingind}{si_1} % 1-singularity index
\DMO{\shyp}{svh} % singularity-hypergraph
\DMO{\sdph}{ssh} % sDP-hypergraph (hypergraph of singular sets)
\DMO{\msdph}{mss} % maximal singular sets
\DMO{\osdph}{ssh_1} % hypergraph of 1-singular sets
\DMO{\mosdph}{mss_1} % maximal 1-singular set
\usepackage{theorem} % am 18.5.2001 von ``Basis'' hierhin verschoben
\usepackage[driverfallback=hypertex]{hyperref}
\nc{\bm}{\boldmath}
\nc{\bmm}[1]{\mbox{\bm$\DST #1$}}% Math. Fettdruck im Text- wie im math. Modus
\nc{\mi}[1]{\bmm{\mathrm{(#1):}} \quad}% math item (verschoben am 28.1.2003)
\DMO{\saturate}{S}

\nc{\Esmusat}{\mc{E}\Smusat} % eventually saturated
\nc{\Eclash}{\mc{E}\Clash} % eventually clashing

\nc{\Dt}[1]{\mc{F}_{#1}} % the saturated clause-sets of deficiency 2 are the $\Dt{n}$

\DMO{\mutt}{nst_2} % MU(2)-type ("non-singularity type")

\renewcommand{\shyp}{S} % in this report we use an older version, the same as in \cite{KullmannZhao2012ConfluenceC}

\newenvironment{keywords}{
  \list{}{\advance\topsep by0.35cm\relax\small
    \leftmargin=1cm
    \labelwidth=0.35cm
    \listparindent=0.35cm
    \itemindent\listparindent
    \rightmargin\leftmargin}\item[\hskip\labelsep
  \bfseries Keywords:]}
{\endlist}

\begin{document}

\pagestyle{headings}

\title{On Davis-Putnam reductions for\\ minimally unsatisfiable clause-sets}

\author{Oliver Kullmann\inst{1} and Xishun Zhao\inst{2}\fnmsep\thanks{Partially supported by NSFC Grant 61272059 and MOE grant 11JJD720020.}}
\institute{
  Computer Science Department, Swansea University, UK\\
  \url{http://cs.swan.ac.uk/~csoliver}\\[0.5ex]
  \and
  Institute of Logic and Cognition\\
  Sun Yat-sen University, Guangzhou, 510275, P.R.C.
}

\maketitle

\begin{abstract}
  DP-reduction $F \leadsto \dpi{v}(F)$, applied to a clause-set $F$ and a variable $v$, replaces all clauses containing $v$ by their resolvents (on $v$). A basic case, where the number of clauses is decreased (i.e., $c(\dpi{v}(F)) < c(F)$), is \emph{singular DP-reduction} (sDP-reduction), where $v$ must occur in one polarity only once. For minimally unsatisfiable $F \in \Musat$, sDP-reduction produces another $F' := \dpi{v}(F) \in \Musat$ with the same deficiency, that is, $\delta(F') = \delta(F)$; recall $\delta(F) = c(F) - n(F)$, using $n(F)$ for the number of variables. Let $\sdp(F)$ for $F \in \Musat$ be the set of results of complete sDP-reduction for $F$; so $F' \in \sdp(F)$ fulfil $F' \in \Musat$, are \emph{nonsingular} (every literal occurs at least twice), and we have $\delta(F') = \delta(F)$. We show that for $F \in \Musat$ all complete reductions by sDP must have the same length, establishing the \emph{singularity index} of $F$. In other words, for $F', F'' \in \sdp(F)$ we have $n(F') = n(F'')$. In general the elements of $\sdp(F)$ are not even (pairwise) isomorphic. Using the fundamental characterisation by Kleine B\"uning, we obtain as application of the singularity index, that we have \emph{confluence modulo isomorphism} (all elements of $\sdp(F)$ are pairwise isomorphic) in case $\delta(F) = 2$. In general we prove that we have confluence (i.e., $\abs{\sdp(F)} = 1$) for saturated $F$ (i.e., $F \in \Smusat$). More generally, we show confluence modulo isomorphism for \emph{eventually saturated} $F$, that is, where we have $\sdp(F) \sse \Smusat$, yielding another proof for confluence modulo isomorphism in case of $\delta(F) = 2$.
\end{abstract}


\tableofcontents

\section{Introduction}
\label{sec:intro}

Minimally unsatisfiable clause-sets (``MU's'') are a fundamental form of irredundant unsatisfiable clause-sets. Regarding the subset relation, they are the hardest examples for proof systems. A substantial amount of insight has been gained into their structure, as witnessed by the handbook article \cite{Kullmann2007HandbuchMU}. A related area of MU, which gained importance in recent industrial applications, is the study of ``MUS's'', that is minimally unsatisfiable \emph{sub-}clause-sets $F' \in \Musat$ with $F' \sse F$ as the ``cores'' of unsatisfiable clause-sets $F$; see \cite{MarquesSilva2012MUS} for a recent overview. For the investigations of this paper there are two main sources: The structure of MU (see Subsection \ref{sec:investmusatk}), and the study of DP-reduction as started with \cite{Ku92,KuLu97,KuLu98}:
\begin{itemize}
\item A fundamental result shown there is that DP-reduction is commutative modulo subsumption (see Subsection \ref{sec:itDPr} for the precise formulation).
\item Singular DP-reduction is a special case of length-reducing DP-reduction (while in general one step of DP-reduction can yield a quadratic blow-up).
\item Confluence \emph{modulo isomorphism} was shown in \cite{Ku92} (Theorem 13, Page 52) for a combination of subsumption elimination with special cases of length-reducing DP-reductions, namely DP-reduction in case no (non-tautological) resolvent is possible, and singular DP-reduction in case there is only one side clause, or the main clause is of length at most $2$ (see Definition \ref{def:singV1}).
\end{itemize}
 The basic questions for this paper are:
 \begin{itemize}
 \item When does singular DP-reduction, applied to MU, yield unique (non-singular) results (i.e., we have confluence)?
 \item And when are the results at least determined up to isomorphism (i.e., we have confluence modulo isomorphism)?
 \end{itemize}
 Different from the result from  \cite{Ku92} mentioned above, we do not consider restricted versions of singular DP-reduction, but we restrict the class of clause-sets to which singular DP-reduction is applied (namely to subclasses of MU).

\subsection{Investigations into the structure of $\Musat(k)$}
\label{sec:investmusatk}

We give now a short overview on the problem of classifying $F \in \Musat$ in terms of the deficiency $\delta(F) := c(F) - n(F)$, that is, the problem of characterising the levels $\Musati{\delta=k} := \set{F \in \Musat : \delta(F) = k}$ (due to greater expressivity and generality, we prefer this notation over $\Musat(k)$); see \cite{Kullmann2007HandbuchMU} for further information.

The field of the combinatorial study of minimally unsatisfiable clause-sets was opened by \cite{AhLi86}, showing the fundamental insight $\delta(F) \ge 1$ for $F \in \Musat$ (see \cite{Ku00f,Kullmann2007HandbuchMU} for generalisations of the underlying method, based on autarky theory). Also $\Smusati{\delta=1}$ was characterised there, where $\Smusat \subset \Musat$ is the set of ``saturated'' minimally unsatisfiable clause-sets, which are minimal not only w.r.t.\ having no superfluous clauses, but also w.r.t.\ that no clause can be further weakened. The fundamental ``saturation method'' $F \in \Musat \leadsto F' \in \Smusat$ was introduced in \cite{FlRe94} (see Definition \ref{def:saturation}). Basic for all studies of MU is detailed knowledge on minimal number of occurrences of a (suitable) variable (yielding a suitable splitting variable): see \cite{KullmannZhao2011Bounds} for the current state-of-art. The levels $\Musati{\delta=k}$ are decidable in polynomial time by \cite{FKS00,Ku99dKo}; see \cite{Szei2002FixedParam,Kullmann2007ClausalFormZI} for further extensions.

``Singular'' variables $v$ in $F \in \Musat$, that is, variables occurring in at least one polarity only once, play a fundamental role --- they are degenerations which (usually) need to be eliminated by \emph{singular DP-reduction}. Let $\Musatns \subset \Musat$ be the set of non-singular minimally unsatisfiable clause-sets (not having singular variables), that is, the results of applying singular DP-reduction to the elements of $\Musat$ as long as possible. The fundamental problem is the characterisation of $\Musatnsi{\delta=k}$ for arbitrary $k \in \NN$. Up to now only $k \le 2$ has been solved: $\Musatnsi{\delta=1}$ has been determined in \cite{DDK98}, while $\Musatnsi{\delta=2} = \Smusatnsi{\delta=2}$ has been determined in \cite{KleineBuening2000SubclassesMU}. Regarding higher deficiencies, until now only (very) partial results in \cite{XD99} exist. Regarding singular minimally unsatisfiable clause-sets, also $\Musati{\delta=1}$ is very well known (with further extensions and generalisations in \cite{Ku99dKo}, and generalised to non-boolean clause-sets in \cite{Kullmann2007ClausalFormZII}), while for $\Musati{\delta=2}$ not much is known (Section \ref{sec:appmu2} provides first insights).

For characterising $\Musatnsi{\delta=k}$, we need (very) detailed insights into (arbitrary) $\Musati{\delta<k}$, since the basic method to investigate $F \in \Musatnsi{\delta=k}$ is to split $F$ into smaller parts from $\Musati{\delta<k}$ (usually containing singular variables). Assuming that we know $\Musatnsi{\delta<k}$, such insights can be based on some classification of $F \in \Musati{\delta<k}$ obtained from the set $\sdp(F) \sse \Musatnsi{\delta<k}$ of singular-DP-reduction results. The easiest case is when $\abs{\sdp(F)} = 1$ holds (confluence), the second-easiest case is where all elements of $\sdp(F)$ are pairwise isomorphic. This is the basic motivation for the questions raised and partially solved in this article. For general $k$ we have no conjecture yet how the classification of $\Musatnsi{\delta=k}$ could look like (besides the basic conjecture that enumeration of the isomorphism types can be done efficiently). However for unsatisfiable hitting clause-sets (two different clauses clash in at least one variable) we have the conjecture stated in \cite{KullmannZhao2011Bounds}, that for every $k \in \NN$ there are only finitely many isomorphism types in $\Uclashnsi{\delta=k}$ (unsatisfiable non-singular hitting clause-sets of deficiency $k$).

\subsection{Overview on results}
\label{sec:overview}

Section \ref{sec:defsing} introduces the basic notions regarding singularity, and the basic characterisations of singular DP-reduction on minimally unsatisfiable clause-sets are given in Subsection \ref{sec:singdp}. In Section \ref{sec:confluence} we consider the question of confluence of singular DP-reduction, with the first main result Theorem \ref{thm:confsat}, showing confluence for saturated clause-sets. Section \ref{sec:permDPred} mainly considers the question of changing the order of DP-reductions without changing the result. The second main result of this article is Theorem \ref{thm:singind}, establishing the singularity index. Section \ref{sec:confmodiso} is devoted to show confluence modulo isomorphism on eventually saturated clause-sets (Theorem \ref{lem:isonssmu}), the third main result. As an application we determine the ``types'' of (possibly singular) minimally unsatisfiable clause-sets of deficiency $2$ via Theorem \ref{thm:confmodisomu2} (Section \ref{sec:appmu2}). We conclude with a collection of open problems in Section \ref{sec:open}.

\subsection{Remarks on related publications}
\label{sec:remarkspub}

The conference-version of this report is \cite{KullmannZhao2012ConfluenceC}:
\begin{enumerate}
\item The report at hand (arXiv:1202.2600), in version 4 or later, contains various proofs, examples and additional results elided in \cite{KullmannZhao2012ConfluenceC}.
\item Additionally two technical mistakes in \cite{KullmannZhao2012ConfluenceC} have been corrected; see Theorem \ref{thm:nbesp} and remarks and Corollary \ref{cor:all2} and remarks.
\end{enumerate}
The journal-version of this report is \cite{KullmannZhao2012ConfluenceJ}, based on version 5 of the report at hand.

\subsection{Applications}
\label{sec:introapp}

Our current main application, which motivated the questions tackled in this paper in the first place, is the project of classifying the structure of $\Musati{\delta=k}$ as discussed in Subsection \ref{sec:investmusatk}: Knowing some form of invariance of singular DP-reduction enables one to classify also \emph{singular} minimally unsatisfiable clause-sets, based on knowing the \emph{non-}singular minimally unsatisfiable clause-sets of the same deficiency; see Section \ref{sec:appmu2} for a first example.

For worst-case upper bounds of SAT decision (or related problems) we sometimes need to guarantee that certain reductions will yield a certain decrease in some parameter, for example the number of variables, independently of the special order of reductions --- this is exactly established for singular DP-reduction by the singularity index (using Corollary \ref{cor:singindsamen}).

Finally, singular DP-reduction is a very basic and efficient reduction, which should be helpful in the search for MUS's, using that a singular variable for $F$ is also singular for $F' \sse F$ with $F' \in \Musat$. The basic results of Section \ref{sec:defsing} make it possible to control the effects of singular DP-reduction, while our main results enable one to estimate the inherent non-determinism. We are aware of the following algorithms using sDP-reduction:
\begin{itemize}
\item A special case of singular DP-reduction, namely unit-clause propagation, has been exploited in \cite{MinLiManyaMohamedouPlanes2010Maxsat} for searching for (some) MUS's; see Subsection \ref{sec:unitp} for further remarks. Note that in the general situation $F' \sse F$ with $F' \in \Musat$, a singular variable for $F'$ might not be singular for $F$ (and thus might go unnoticed) --- the problem is that we don't know $F'$ in advance. However in the case of unit-clauses $\set{x} \in F$ we can discard all clauses $C \in F$ with $\set{x} \subset C$ (for a MUS involving $\set{x}$), and so the singular literal $x$ won't be missed.
\item DP-reduction in general has been used in theoretical as well as in practical SAT-algorithms:
  \begin{enumerate}
  \item \cite{DP60} used DP-reductions for (complete) SAT solving, by unrestricted application of the reduction rule.
  \item In \cite{Franco1991InfrequentVariables} a simple case of DP-reduction, namely considering only variables occurring at most twice, has been analysed probabilistically.
  \item DP-reductions has been used in the worst-case analysis of algorithms in \cite{Ku92,KuLu97,KuLu98}; especially in \cite{KuLu97,KuLu98} it is shown that allowing reductions $F \leadsto \dpi{v}(F)$ with up to $K$ new clauses for a fixed $K$, i.e., $c(\dpi{v}(F)) \le c(F) + K$, can improve worst-case performance.
  \item In \cite{Gelder2001ResSATsolver} this DP-reduction with bounded clause-number-increase has been used at each node of the search tree of a SAT solver, with $K \approx 200$.
  \item In \cite{SubbarayanPradhan2004LDP} another criterion analysed in \cite{Ku92,KuLu97,KuLu98}, namely $\ell(\dpi{v}(F)) \le \ell(F)$ has been implemented, where $\ell(F) := \sum_{C \in F} \abs{C}$ is the number of literal occurrences, this time as a free-standing preprocessor. Singular DP-reduction is not covered by this criterion (since the number of literal-occurrences can be increased by sDP-reduction).
  \item This approach has been further developed in \cite{EenBiere2005Satelite}, but now using $K = 0$, i.e., $c(\dpi{v}(F)) \le c(F)$. Again a free-standing preprocessor has been provided, called ``satELite''. Now sDP-reduction is covered.
  \end{enumerate}
  This preprocessor was incorporated into several recent SAT solvers, most notably into the \texttt{minisat} solvers from version 2.0 on. So a ``minimal unsatisfiable core (or subset) extraction'' algorithm like \texttt{Haifa-MUC}, the winner of the SAT 2011 competition regarding this task, applies sDP-reduction.
\end{itemize}

\section{Preliminaries}
\label{sec:prelim}

We follow the general notations and definitions as outlined in \cite{Kullmann2007HandbuchMU}. We use $\NN = \set{1,2,\dots}$ and $\NNZ = \NN \cup \set{0}$.

Consider a relation $R \sse X^2$ on a set $X$; for us typically $X$ is the set $\Cls$ of all clause-sets or the set $\Musat$ of all minimally unsatisfiable clause-sets. We view $R$ as a ``reduction'', and we write $x \leadsto x'$ for $(x,x') \in R$. Such a reduction is called \emph{terminating} if there are no infinite chains $x_1 \leadsto x_2 \leadsto x_3 \leadsto \dots$ of reductions. Using the reflexive-transitive closure $\leadsto^*$ (that is, zero, one or more reductions taking place), for a terminating reduction and every $x \in X$ there is at least one $x' \in X$ with $x \leadsto^* x'$ such that there is no $x'' \in X$ with $x' \leadsto x''$. A terminating reduction is called \emph{confluent} if this $x'$ is always unique.

An example for a terminating and confluent reduction-relation is unrestricted DP-reduction $F \leadsto \dpi{v}(F)$ for a clause-set $F \in \Cls$ and a variable $v \in \var(F)$, as defined below.

\subsection{Clause-sets}

The (infinite) set of all variables is $\Va$, while the set of all literals is $\Lit$, where we identify the positive literals with variables, that is, we assume $\Va \subset \Lit$. Complementation is an involution of $\Lit$, and is denoted for literals $x \in \Lit$ by $\ol{x} \in \Lit$. For a set $L$ of literals we define $\ol{L} := \set{\ol{x} : x \in L}$ (so $\Lit$ is the disjoint union of $\Va$ and $\ol{\Va}$). A \textbf{clause} $C$ is a finite and clash-free set of literals (i.e., $C \cap \ol{C} = \es$), while a \textbf{clause-set} $F \in \Cls$ is a finite set of clauses. The empty clause is denoted by $\bmm{\bot} := \es$, and the empty clause-set is denoted by $\bmm{\top} \in \Cls$. We denote by $\var(F)$ the set of (occurring) variables, by $n(F) := \abs{\var(F)}$ the number of variables, by $c(F) := \abs{F}$ the number of clauses, and finally by $\delta(F) := c(F) - n(F)$ the deficiency. For clause-sets $F, G$ we denote by \bmm{F \cong G} that both clause-sets are \textbf{isomorphic}, that is, the variables of $F$ can be renamed and potentially flipped so that $F$ is turned into $G$; more precisely, an isomorphism $\alpha$ from $F$ to $G$ is a bijection $\alpha$ on literal-sets which preserves complementation and which maps the clauses of $F$ precisely to the clauses of $G$. The \emph{literal-degree} $\ldeg_F(x) \in \NNZ$ of a literal $x$ for a clause-set $F$ is the number of clauses the literal appears in, i.e., $\ldeg_F(x) := \abs{\set{C \in F : x \in C}}$. The \emph{variable-degree} $\vdeg_F(v) \in \NNZ$ for a variable $v$ is the number of clauses the variable appears in, i.e., $\vdeg_F(v) := \ldeg_F(v) + \ldeg_F(\ol{v})$.

For a clause-set $F$ and a variable $v$, by \bmm{\dpi{v}(F)} we denote the result of applying DP-reduction on $v$ (``DP'' stands for ``Davis-Putnam'', who introduced this operation in \cite{DP60}), that is, removing all clauses containing $v$ and adding all resolvents on $v$. More formally
\begin{displaymath}
  \dpi{v}(F) := \set{C \in F : v \notin \var(C)} \cup \set{C \res D : C, D \in F, \, C \cap \ol{D} = \set{v}},
\end{displaymath}
where clauses $C, D$ are resolvable iff they clash in exactly one literal, i.e., iff $\abs{C \cap \ol{D}} = 1$, while for resolvable clauses $C, D$ the resolvent $\bmm{C \res D} := (C \cup D) \sm \set{x,\ol{x}}$ for $C \cap \ol{D} = \set{x}$ is defined as the union minus the resolution literals (the two clashing literals). $\dpi{v}(F)$ is logically equivalent to the existential quantification of $F$ by $v$, and thus $F$ and $\dpi{v}(F)$ are satisfiability-equivalent, that is, $\dpi{v}(F)$ is satisfiable iff $F$ is satisfiable.

We can define $\Sat \subset \Cls$, the set of all satisfiable clause-set, as the set of $F \in \Cls$ where reduction by DP will finally yield $\top$, the empty clause-set, while we can define $\Usat = \Cls \sm \Sat$, the set of all unsatisfiable clause-set, as the set of $F \in \Cls$ where reduction by DP will finally yield $\set{\bot}$, the clause-set consisting of the empty clause.

Since DP-reduction on $v$ removes at least variable $v$, every sequence of applications of DP until no variables are left must end up either in $\top$ or in $\set{\bot}$. The sa\-tis\-fi\-abi\-li\-ty-in\-var\-ian\-ce of DP-reduction yields that the final result does not depend on the choices involved, but only on the satisfiability resp.\ unsatisfiability of the starting clause-set. So unrestricted DP-reduction is terminating and confluent; a proof of confluence from first principles (by combinatorial means) is achieved by Lemma \ref{lem:DPcomm}.

\subsection{Minimal unsatisfiability}

The set of minimally unsatisfiable clause-sets is $\Musat \subset \Usat$, the set of all clause-sets which are unsatisfiable, while removal of any clause makes them satisfiable. Furthermore the set of saturated minimally unsatisfiable clause-sets is $\Smusat \subset \Musat$, which is the set of minimally unsatisfiable clause-sets such that addition of any literal to any clause renders them satisfiable. Note that for $v \in \var(F)$ with $F \in \Musat$ we have $\vdeg_F(v) \ge 2$. We recall the fact (\cite{FlRe94} and Lemma 5.1 in \cite{Kullmann2007ClausalFormZII}) that every minimally unsatisfiable clause-set $F \in \Musat$ can be \textbf{saturated}, i.e., by adding literal occurrences to $F$ we obtain $F' \in \Smusat$ with $\var(F') = \var(F)$ such that there is a bijection $\alpha: F \ra F'$ with $C \sse \alpha(C)$ for all $C \in F$. The details are as follows.
\begin{definition}\label{def:saturation}
  The operation $\bmm{\saturate(F,C,x)} := (F \sm \set{C}) \cup (C \cup \set{x}) \in \Cls$ (adding literal $x$ to clause $C$ in $F$) is defined if $F \in \Cls$, $C \in F$, and $x$ is a literal with $\var(x) \in \var(F) \sm \var(C)$.
  A \textbf{saturation} $F' \in \Smusat$ of $F \in \Musat$ is obtained by a sequence $F = F_0, \dots, F_m = F'$, $m \in \NNZ$,
  \begin{itemize}
  \item such that for $0 \le i < m$ there are $C_i, x_i$ with $F_{i+1} = \saturate(F_i,C_i,x_i)$,
  \item such that for all $1 \le i \le m$ we have $F_i \notin \Sat$,
  \item and such that the sequence cannot be extended.
  \end{itemize}
  Note that $n(F') = n(F)$ and $c(F') = c(F)$ holds (and thus $\delta(F') = \delta(F)$). More generally, a \textbf{partial saturation} of a clause-set $F \in \Musat$ is a clause-set $F' \in \Musat$ such that $\var(F') = \var(F)$ and there is a bijection $\alpha: F \ra F'$ such that for all $C \in F$ we have $C \sse \alpha(C)$.
\end{definition}
Please note that if for $F \in \Musat$ and $F' := \saturate(F,C,x)$ we have $F' \notin \Sat$, then actually $F' \in \Musat$ must hold. Thus if $F'$ is a saturation of $F \in \Musat$ in the sense of Definition \ref{def:saturation}, then actually $F'$ is saturated (minimally unsatisfiable).

A clause-set $F$ is \textbf{hitting} if every two different clauses clash in at least one literal. The set of hitting clause-sets is denoted by
\begin{displaymath}
  \bmm{\Clash} := \set{F \in \Cls \mb \fa\, C, D \in F, C \not= D : C \cap \ol{D} \not= \es} \subset \Cls,
\end{displaymath}
the set of unsatisfiable hitting clause-sets by $\bmm{\Uclash} := \Clash \cap \Usat$. When interpreting $F$ as DNF, hitting clause-sets are known as ``disjoint'' or ``orthogonal'' DNF; see Chapter 7 in \cite{CramaHammer2011BooleanFunctions}.

\begin{lemma}\label{lem:UHITSMU}
  We have $\Uclash \subset \Smusat$.
\end{lemma}
\begin{proof}
For $F \in \Clash$ we have $F \in \Usat$ iff $\sum_{C \in F} 2^{-\abs{C}} = 1$ (see \cite{Kullmann2007HandbuchMU}; the point is that two clashing clauses do not have a common falsifying assignment). Thus adding a literal to a clause of $F \in \Uclash$ makes $F$ satisfiable. See Example \ref{exp:F2F3} for an example showing that the inclusion is strict. \qed
\end{proof}

\begin{example}\label{exp:F2F3}
  Two unsatisfiable hitting clause-sets used in various examples are:
  \begin{eqnarray*}
    \Dt{2} & := & \set{\set{v_1,v_2},\set{\ol{v_1},\ol{v_2}},\set{\ol{v_1},v_2},\set{\ol{v_2},v_1}}\\
    \Dt{3} & := & \set{\set{v_1,v_2,v_3},\set{\ol{v_1},\ol{v_2},\ol{v_3}},\set{\ol{v_1},v_2},\set{\ol{v_2},v_3},\set{\ol{v_3},v_1}}.
  \end{eqnarray*}
  And an example for an element of $\Smusat \sm \Uclash$ is given by
  \begin{displaymath}
    \Dt{4} := \set{\set{v_1,v_2,v_3,v_4},\set{\ol{v_1},\ol{v_2},\ol{v_3},\ol{v_4}},\set{\ol{v_1},v_2},\set{\ol{v_2},v_3},\set{\ol{v_3},v_4},\set{\ol{v_4},v_1}}.
  \end{displaymath}
  To see $\Dt{4} \in \Smusat$ it is easiest to use Corollary 5.3 in \cite{Kullmann2007ClausalFormZII}, that is, we have to show that for all $v \in \var(\Dt{4})$ and $\ve \in \set{0,1}$ we have $\pao v{\ve} * \Dt{4} \in \Musat$. W.l.o.g.\ $v = v_1$ and $\ve = 0$, and then $\pao v{\ve} * \Dt{4} = \set{\set{v_2,v_3,v_4}, \set{\ol{v_2},v_3},\set{\ol{v_3},v_4},\set{\ol{v_4}}} \in \Musati{\delta=1}$. The clause-sets $\Dt{2}, \Dt{3}, \Dt{4}$ are elements of $\Musati{\delta=2}$; see Section \ref{sec:appmu2} for more on this class.
\end{example}

The following (new) observation is fundamental for the study of hitting clause-sets:
\begin{lemma}\label{lem:hitDPg}
  For $F \in \Clash$ and a variable $v$ we have $\dpi{v}(F) \in \Clash$.
\end{lemma}
\begin{proof}
Consider clauses $E_1, E_2 \in \dpi{v}(F)$, $E_1 \not= E_2$. If $E_1, E_2 \in F$, then $E_1, E_2$ clash since $F$ is hitting. The two remaining cases are (w.l.o.g.) $E_1 \in F, E_2 \notin F$ and $E_1, E_2 \notin F$. In the first case assume $E_2 = C_2 \res D_2$ for $C_2, D_2 \in F$ with $C_2 \cap \ol{D_2} = \set{v}$. Since $v \notin \var(E_1)$, it clashes $E_1$ with $C_2$ (as well as with $D_2$) and thus with $E_2$. For the second case also assume $E_1 = C_1 \res D_1$ for $C_1, D_1 \in F$ with $C_1 \cap \ol{D_1} = \set{v}$. We must have $C_1 \not= C_2$ or $D_1 \not= D_2$, yielding a clash between $C_1, C_2$ resp.\ $D_1, D_2$, and thus also $E_1, E_2$ clash. \qed
\end{proof}

Since DP-reduction preserves unsatisfiability, we get:
\begin{corollary}\label{cor:hituDPg}
  For $F \in \Uclash$ and a variable $v$ we have $\dpi{v}(F) \in \Uclash$.
\end{corollary}

\section{Singularity}
\label{sec:defsing}

In this section we present basic results on singular variables in minimally unsatisfiable clause-sets. Lemmas \ref{lem:singDpMU}, \ref{lem:singDpSMU} yield basic characterisations of singular DP-reduction for minimally unsatisfiable resp.\ saturated minimally unsatisfiable clause-sets (some of these results were discussed in \cite{KullmannSzeider2003Communication}), while Lemma \ref{lem:fullunit} shows that in the context of MU unit-clause propagation is a special case of singular DP-reduction. These results are straight-forward, but the choice of concepts is important, and the facts are somewhat subtle.

\subsection{Singular variables}
\label{sec:singvar}

\begin{definition}\label{def:singV1}
  We call a variable $v$ \textbf{singular} for a clause-set $F \in \Cls$ if we have $\min(\ldeg_F(v),\ldeg_F(\ol{v})) = 1$; the set of singular variables of $F$ is denoted by $\bmm{\varsing(F)} \sse \var(F)$. $F$ is called \textbf{nonsingular} if $F$ does not contain singular variables. Furthermore we use the following notations:
  \begin{itemize}
  \item $\bmm{\Musatns} := \set{F \in \Musat : \varsing(F) = \es}$ denotes the set of nonsingular MU's;
  \item $\bmm{\Smusatns} := \Smusat \cap \Musatns$ is the set of nonsingular saturated MU's;
  \item $\bmm{\Uclashns} := \Uclash \cap \Smusatns = \Clash \cap \Musatns$ is the set of nonsingular unsatisfiable hitting clause-sets.
  \end{itemize}
  More precisely:
  \begin{itemize}
  \item We call variable $v$ \textbf{\bmm{m}-singular} for $F$ for some $m \in \NN$, if $v$ is singular for $F$ with $m = \vdeg_F(v) - 1$. The set of 1-singular variables of $F$ is denoted by $\bmm{\varosing(F)} := \set{v \in \Va : \ldeg_F(v) = \ldeg_F(\ol{v}) = 1} \sse \varsing(F)$.
  \item A \textbf{non-1-singular} variable is a variable which $m$-singular for some $m \ge 2$ (so ``non-1-singular'' variables are singular). The set of non-1-singular variables of $F$ is denoted by $\bmm{\varnosing(F)} := \varsing(F) \sm \varosing(F)$. 
  \end{itemize}
  A \textbf{singular literal} for a singular variable $v$ is a literal $x$ with $\var(x) = v$ and $\ldeg_F(x) = 1$; if the underlying variable is $1$-singular, then some choice is applied, so that we can speak of ``the'' singular literal of a singular variable. For the singular literal $x$ for $v$ we call the clause $C \in F$ with $x \in C$ the \textbf{main clause}, while the \textbf{side clauses} are the clauses $D_1, \ldots, D_m \in F$ with $\ol{x} \in D_i$ (here $v$ is $m$-singular).
\end{definition}
\begin{example}\label{exp:singvar}
  For $F := \set{\set{a},\set{\ol{a},b},\set{\ol{a},\ol{b}}}$, variable $a$ is $2$-singular, while variable $b$ is $1$-singular, and thus $\varsing(F) = \set{a,b}$, $\varosing(F) = \set{b}$ and $\varnosing(F) = \set{a}$. The main clause of $a$ is $\set{a}$, its side clauses are $\set{\ol{a},b},\set{\ol{a},\ol{b}}$, while for the main clause of $b$ there is the choice between $\set{\ol{a},b}$ and $\set{\ol{a},\ol{b}}$.

  In general, if $F \in \Musat$ contains a unit-clause $\set{x} \in F$, then $\var(x)$ is singular for $F$ (see Lemma \ref{lem:fullunit}). Thus the clause-sets $\set{\bot}$ and $\Dt{2}$ (recall Example \ref{exp:F2F3}) are the two smallest elements of $\Musatns$, $\Smusatns$ and $\Uclashns$ regarding the number of clauses.
\end{example}

\subsection{Singular DP-reduction}
\label{sec:singdp}

The following special application of DP-reduction appears at many places in the literature (see \cite{KleineBuening2000SubclassesMU}, or Appendix B in \cite{Ku99dKo} and subsequent \cite{Szei2002FixedParam,Kullmann2007ClausalFormZI}), and is fundamental for investigations of minimally unsatisfiable clause-sets:
\begin{definition}\label{def:singDP}
  A \textbf{singular DP-reduction} is a reduction $F \leadsto \dpi{v}(F)$, where $v$ is singular for $F \in \Musat$. For $F, F' \in \Musat$ by \bmm{F \tsdp F'} we denote that $F'$ is obtained from $F$ by one step of singular DP-reduction; i.e., there is a singular variable $v$ for $F$ with $F' = \dpi{v}(F)$, where $v$ is called the \textbf{reduction variable}. And we write $F \tsdps F'$ if $F'$ is obtained from $F$ by an arbitrary number of steps (possibly zero) of singular DP-reductions. The set of all nonsingular clause-sets obtainable from $F$ by singular DP-reduction is denoted by \bmm{\sdp(F)}:
  \begin{displaymath}
    \sdp(F) := \set{F' \in \Musatns : F \tsdps F'}.
  \end{displaymath}
\end{definition}

The following lemma is kind of ``folklore'', but apparently the only place where its assertions are (partially) stated in the literature (in a more general form) is \cite{Kullmann2007ClausalFormZI}, Lemma 6.1 (we add here various details):
\begin{lemma}\label{lem:singDpMU}
  Consider a clause-set $F$ and a singular variable $v$ for $F$. Then the following assertions are equivalent:
  \begin{enumerate}
  \item\label{lem:singDpMU1} $F$ is minimally unsatisfiable.
  \item\label{lem:singDpMU2} $\delta(\dpi{v}(F)) = \delta(F)$ and $\dpi{v}(F)$ is minimally unsatisfiable.
  \item\label{lem:singDpMU3} $\dpi{v}(F)$ is minimally unsatisfiable, and for the main clause $C$ and the side clauses $D_1, \dots, D_m$ for $v$ (in $F$) we have:
    \begin{enumerate}
    \item\label{lem:singDpMU3a} Every $D_i$ clashes with $C$ in exactly one variable (namely in $v$).
    \item\label{lem:singDpMU3b} For $1 \le i \ne j \le m$ we have $C \res D_i \not= C \res D_j$.
    \item\label{lem:singDpMU3c} For $E \in F$ with $v \notin \var(E)$ and for all $1 \le i \le m$ we have $C \res D_i \not= E$.
    \end{enumerate}
  \end{enumerate}
\end{lemma}
\begin{proof}
The equivalence of Part \ref{lem:singDpMU1} and Part \ref{lem:singDpMU2} is a special case of Lemma 6.1 in \cite{Kullmann2007ClausalFormZI}.
Part \ref{lem:singDpMU2} implies Part \ref{lem:singDpMU3}, since if one of the conditions \ref{lem:singDpMU3a}, \ref{lem:singDpMU3b} or
\ref{lem:singDpMU3c} would not hold, then the deficiency of $\dpi{v}(F)$ would be (strictly) smaller than $F$, contradicting the assumption $\delta(\dpi{v}(F)) = \delta(F)$.
Finally we show that Part \ref{lem:singDpMU3} implies Part \ref{lem:singDpMU1}. Since $\dpi{v}(F)$ is minimally unsatisfiable, $F$ is unsatisfiable. Now suppose that $F$ is not minimally unsatisfiable. So for some clause $E \in F$ the clause-set $F' := F\setminus \{E\}$ is still unsatisfiable. By condition \ref{lem:singDpMU3a} we know that $C \res D_i$ must be in $\dpi{v}(F)$ for all $i \in \tb 1m$. Thus clause $E$ can not be the main clause $C$, and if $m=1$, then $E$ can not be the side clause neither. So $v$ is still a singular variable in $F'$. Since $\dpi{v}(F)$ is minimally unsatisfiable, while we have $\dpi{v}(F') \subseteq \dpi{v}(F)$, we obtain $\dpi{v}(F')= \dpi{v}(F)$, that is, either $E$ is one of the side clauses and its resolvent with $C$ was obtained by some other resolution or was already present, or $E$ does not contain $v$, and thus $E$ must be a resolvent. In any case we get a contradiction with one of \ref{lem:singDpMU3b} or \ref{lem:singDpMU3c}. \qed
\end{proof}

\begin{corollary}\label{cor:sdppresmu}
  If $F \in \Musat$ and $v$ is a singular variable of $F$, then also $\dpi{v}(F) \in \Musat$, where $\delta(\dpi{v}(F)) = \delta(F)$. So the classes $\Musat_{\delta=k}$ for $k \in \NN$ are stable under singular DP-reduction.
\end{corollary}

\begin{corollary}\label{cor:sdppartsatur}
  Consider $F \in \Musat$ and a singular variable $v$ with singular literal $x$, with main clause $C$ and side clauses $D_1,\dots,D_m$. Then adding $C \sm \set{x}$ to $D_i$ for all $i \in \tb 1m$ is a partial saturation of $F$ (recall Definition \ref{def:saturation}).
\end{corollary}
\begin{proof}
  Let $F'$ be obtained from $F$ by replacing the clauses $D_i$ by the clauses $D_i \cup (C \sm \set{x})$ for each $i \in \tb 1m$ (note that by Lemma \ref{lem:singDpMU}, Part \ref{lem:singDpMU3a}, the literal-sets $D_i \cup (C \sm \set{x})$ are clash-free and thus indeed clauses). By Lemma \ref{lem:singDpMU} we know that $\dpi{v}(F) \in \Musat$ holds. Now $\dpi{v}(F') = \dpi{v}(F)$, and so in order to show that $F' \in \Musat$, we need to show that the three conditions of Part \ref{lem:singDpMU3} of Lemma \ref{lem:singDpMU} hold. Condition \ref{lem:singDpMU3a} holds by definition. And conditions \ref{lem:singDpMU3b}, \ref{lem:singDpMU3c} follow from the fact (which was already used for $\dpi{v}(F') = \dpi{v}(F)$), that the changed clauses $D_i$ yield the same resolvents with clause $C$. \qed
\end{proof}

Lemma \ref{lem:singDpMU} can be strengthened for saturated $F$ by requiring special conditions for the occurrences of the singular variable.
\begin{lemma}\label{lem:singDpSMU}
  Consider a clause-set $F$ and a singular variable $v$ for $F$. For the singular literal $x$ for $v$ consider the main clause $C$ and the side clauses $D_1, \dots D_m \in F$. Let $C' := C \sm \set{x}$ and $D_i' := D_i \sm \set{\ol{x}}$. The following assertions are equivalent:
  \begin{enumerate}
  \item\label{lem:singDpSMU1} $F$ is saturated minimally unsatisfiable.
  \item\label{lem:singDpSMU2} The following three conditions hold:
    \begin{enumerate}
    \item\label{lem:singDpSMU2a} $\dpi{v}(F)$ is saturated minimally unsatisfiable;
    \item\label{lem:singDpSMU2b} $C' = \bca_{i=1}^m D_i'$;
    \item\label{lem:singDpSMU2c} for every $E \in F$ with $v \notin \var(E)$ we have $C' \not\sse E$.
    \end{enumerate}
    Note that conditions \ref{lem:singDpSMU2b}, \ref{lem:singDpSMU2c} together imply the condition that for $E \in F$ we have $C' \sse E$ if and only if $v \in \var(E)$ holds.
  \end{enumerate}
\end{lemma}
\begin{proof}
First assume that $F$ is saturated minimally unsatisfiable. If there would be $E \in F$ with $v \notin \var(E)$ and $C' \sse E$, then for $F' := \saturate(F,E,\ol{v})$ we had $\dpi{v}(F') = \dpi{v}(F)$, and thus $F'$ would be unsatisfiable, contradicting saturatedness of $F$. We have $C' \sse \bca_{i=1}^m D_i'$, since if there were a literal $y \in C'$ and $y \notin D_i'$ for some $i$, then $\dpi{v}(\saturate(F,D_i,y)) = \dpi{v}(F)$. And we have $C' \supseteq \bca_{i=1}^m D_i'$, since if there were a literal $y$ contained in all $D_i'$, but not in $C'$, then $\dpi{v}(\saturate(F,C,y)) = \dpi{v}(F)$.

By Lemma \ref{lem:singDpMU} we know that $\dpi{v}(F)$ is minimally unsatisfiable, and that all resolutions are carried out, with no contraction due to coinciding resolvents or coincidence of a resolvent with an existing clause. Assume that $\dpi{v}(F)$ is not saturated, that is, there is a clause $E$ and a literal $y$ with $G := \saturate(\dpi{v}(F),E,y) \in \Usat$. If $E \in F$ then $\dpi{v}(\saturate(F,E,y)) = G \in \Usat$, and so there is some $1 \le i \le m$ with $E = C \res D_i$. But now $\dpi{v}(\saturate(F,D_i,y)) = G$, yielding a contradiction.

Now we consider the opposite direction, that is, we assume that $C' = \bca_{i=1}^m D_i'$, that $\dpi{v}(F)$ is saturated minimally unsatisfiable, and that $C'$ is contained in some clause of $F$ iff this clause contains the variable $v$. First we establish the three conditions from Lemma \ref{lem:singDpMU}, Part \ref{lem:singDpMU3}. Since clauses are clash-free, $C'$ has no conflict with any $D_i'$, and thus the clash-freeness-condition is fulfilled. If we had $C \res D_i = C \res D_j$ for $i \not= j$, then w.l.o.g.\ there must be a literal $y \in C'$ with $y \in D_i'$ and $y \notin D_j'$, which is impossible since $C'$ contains only literals which are common to all side clauses. Finally, since all resolvents $C \res D_i$ subsume the parent clause $D_i$, by the minimal unsatisfiability of $F$ also Condition \ref{lem:singDpMU3c} is fulfilled. So we have established that $F$ is minimally unsatisfiable.

Assume that $F$ is not saturated, that is, there exists a clause $E \in F$ and a literal $y$ with $G := \saturate(F,E,y) \in \Musat$. Let $F' := \dpi{v}(F)$ and $G' := \dpi{v}(G)$ (note $G' \in \Usat$, and that $F' \in \Smusat$ by assumption). Our strategy is to derive a contradiction by showing that literal occurrences can be added to $F'$ in such a way that $G'$ is obtained, contradicting that $F'$ is saturated.

First consider $E \notin \set{C} \cup \set{D_i}_{1 \le i \le m}$. If $\var(y) \not= v$, then $G' = \saturate(F',E,y)$. If $y = \ol{v}$, then $G' = \saturate(F', \set{E}, C')$ (using Condition \ref{lem:singDpSMU2c}). It remains the case $y = v$, but this case is impossible since then for all $1 \le i \le m$ we have $C \res D_i = D_i' \sse (E \cup \set{v}) \res D_i = E \cup D_i'$, and thus $\dpi{v}(G)$ would be satisfiability equivalent to $\dpi{v}(G \sm \set{E \cup \set{v}})$, whence $G$ would not be minimally unsatisfiable.

So we have $E \in \set{C} \cup \set{D_i}_{1 \le i \le m}$, i.e., $v \in \var(C)$. If $E = C$, then $G' = \saturate(F', \set{D_i'}_{1 \le i \le m},y)$, using that $C'$ is the intersection all the $D_i'$, and thus at least one $D_i'$ does not contain $y$. And if $E = C_i$ for some $i$, then $G' = \saturate(F', D_i', y)$. \qed
\end{proof}

\begin{corollary}\label{cor:sdppressmu}
  The class $\Smusat$ is stable under singular DP-reduction.
\end{corollary}

\subsection{Unit-clauses}
\label{sec:unitp}

In this subsection we explore the observation that unit-clause propagation for minimally unsatisfiable clause-sets is a special case of singular DP-reduction. First we show that unit-clauses in minimally unsatisfiable clause-sets can be considered as special cases of singular variables in the following sense:
\begin{lemma}\label{lem:fullunit}
  Consider $F \in \Musat$.
  \begin{enumerate}
  \item\label{lem:fullunit1} If $v$ is singular for $F$ and occurs in every clause of $F$ (positively or negatively), then we have $\set{v} \in F$ or $\set{\ol{v}} \in F$.
  \item\label{lem:fullunit2} If $\set{x} \in F$ for some literal $x$, then $v := \var(x)$ is singular in $F$ (with $\ldeg_F(x) = 1$). If here $F$ is saturated, then $v$ must occur in every clause of $F$.
  \end{enumerate}
\end{lemma}
\begin{proof}
For Part \ref{lem:fullunit1} consider a main clause $C$ for $v$, and assume w.l.o.g.\ $v \in C$. Since every other clause $D \in F \sm \set{C}$ contains $\ol{v}$, while $C$ has exactly one clash with $D$ by Lemma \ref{lem:singDpMU}, Part \ref{lem:singDpMU3a}, literals in $C \sm \set{v}$ are pure in $F$, and thus there can not be any (that is, $C = \set{v}$ holds), since $F$ is minimally unsatisfiable. For Part \ref{lem:fullunit2} we first observe that every other clause of $F$ containing $x$ would be subsumed by $\set{x}$, which is impossible since $F$ is minimally unsatisfiable. If $F$ is saturated, then every clause $D \in F \sm \set{\set{x}}$ must contain $\ol{x}$ by Lemma \ref{lem:singDpSMU}, Part \ref{lem:singDpSMU2c}. \qed
\end{proof}
So nonsingular minimally unsatisfiable clause-sets do not contain unit-clauses.
\begin{example}\label{exp:upsing}
  Some examples illustrating the relation between unit-clauses and singular variables for $\Musat$:
  \begin{enumerate}
  \item From $\Dt{2} = \set{\set{v_1,v_2},\set{\ol{v_1},\ol{v_2}},\set{\ol{v_1},v_2},\set{\ol{v_2},v_1}} \in \Uclashnsi{\delta=2}$ (recall Example \ref{exp:F2F3}) we obtain, using ``inverse unit-clause elimination'':
    \begin{enumerate}
    \item $\set{\set{x},\set{v_1,v_2,\ol{x}},\set{\ol{v_1},\ol{v_2},\ol{x}},\set{\ol{v_1},v_2,\ol{x}},\set{\ol{v_2},v_1,\ol{x}}} \in \Uclashi{\delta=2}$
    \item $\set{\set{x},\set{v_1,v_2,\ol{x}},\set{\ol{v_1},\ol{v_2},\ol{x}},\set{\ol{v_1},v_2,\ol{x}},\set{\ol{v_2},v_1}} \in \Musati{\delta=2} \sm \Smusati{\delta=2}$.
    \end{enumerate}
  \item $\set{\set{a,b},\set{a,\ol{b}},\set{\ol{a},c},\set{\ol{a},\ol{c}}} \in \Uclashi{\delta=1}$ contains the two singular variables $b,c$, while not containing a unit-clause.
  \end{enumerate}
\end{example}
If $F \in \Musat_{\delta=k}$ contains a unit-clause $\set{x} \in F$, then we can apply singular DP-reduction for the underlying variable of $x$, and the result $\dpi{\var(x)}(F) \in \Musati{\delta=k}$ is the same as the result of the usual unit-clause elimination for $\set{x}$ (setting $x$ to true, and simplifying accordingly). We now consider the case where repeated unit-clause elimination, i.e., unit-clause propagation, yields the empty clause.

In \cite{DDK98} it has been shown that for minimally unsatisfiable clause-sets $F \in \Musat$ the following properties are equivalent:
\begin{enumerate}
\item $F$ can be reduced by sDP to $\set{\bot}$, i.e., $F \tsdps \set{\bot}$.
\item All sDP-reductions of $F$ end with $\set{\bot}$, i.e., $\sdp(F) = \set{\bot}$.
\item $\delta(F) = 1$.
\end{enumerate}
Let $\rk_1: \Cls \ra \Cls$ denote unit-clause propagation, that is, $\rk_1(F) := \set{\bot}$ if $\bot \in F$, $\rk_1(F) := F$ if all clauses of $F$ have length at least two, and otherwise $\rk_1(F) := \rk_1(\pao x1 * F)$ for $\set{x} \in F$, where $\pao x1 * F$ means setting literal $x$ to true, i.e., removing clauses containing $x$, and removing literal $\ol{x}$ from the remaining clauses (see \cite{Ku99b,Ku00g} for a proof of confluence, i.e., independence of the choice of the unit-clauses $\set{x}$, and for generalisations). So, if for $F \in \Musat$ we have $\rk_1(F) = \set{\bot}$, then we know $F \in \Musati{\delta=1}$. Now it is well-known (first shown in \cite{He74}) that for $F \in \Cls$ we have $\rk_1(F) = \set{\bot}$ if there is $F' \sse F$ with $F' \in \Musat \cap \Rho$, where $\Rho$ is the class of renamable (or ``hidden'') Horn clause-sets, that is, we have $F' \in \Rho$ iff there is a Horn clause-set $F'' \in \Ho$ with $F' \cong F''$, where $\Ho := \set{F \in \Cls \mb \fa\, C \in F : \abs{C \cap \Va} \le 1}$ (each clause contains at most one positive literal). Altogether follows the following well-known characterisation:

\begin{lemma}\label{lem:characupecmu}
  For $F \in \Musat$ holds $\rk_1(F) = \set{\bot}$ iff $F \in \Musati{\delta=1} \cap \Rho$.
\end{lemma}
Reconstruction of minimally unsatisfiable sub-clause-sets $F' \sse F \in \Cls$ in case of $\rk_1(F) = \set{\bot}$ is performed in \cite{MinLiManyaMohamedouPlanes2010Maxsat}, in the context of MAXSAT solving, and by Lemma \ref{lem:characupecmu} we have $F' \in \Musati{\delta=1} \cap \Rho$ for these $F'$. In \cite{MinLiManyaMohamedouPlanes2010Maxsat} also failed-literal elimination is discussed, i.e., the case $\rk_2(F) = \set{\bot}$ (see \cite{Ku99b,Ku00g}), where $\rk_2: \Cls \ra \Cls$ is defined as $\rk_2(F) := \rk_2(\pao x1 * F)$ for a literal $x$ with $\rk_1(\pao x0 * F) = \set{\bot}$, while otherwise $\rk_2(F) := F$.

\begin{example}\label{exp:r2sdp}
   The following examples show that $\rk_2$ and sDP-reduction are incomparable regarding derivation of a contradiction:
   \begin{enumerate}
   \item $\Dt{2} \in \Uclashnsi{\delta=2}$ (recall Example \ref{exp:F2F3}) has $\rk_2(\Dt{2}) = \set{\bot}$.
   \item $F \!:=\! \set{\!\set{a,b,c},\set{a,b,\ol{c}},\set{a,\ol{b},d},\set{a,\ol{b},\ol{d}},\set{\ol{a},e,f},\set{\ol{a},e,\ol{f}},\set{\ol{a},\ol{e},g},\set{\ol{a},\ol{e},\ol{g}}\!}$ fulfils $F \in \Uclashi{\delta=1}$, while $\rk_2(F) = F$ (all clauses of $F$ have length $3$).
   \end{enumerate}
\end{example}

\section{Confluence of singular DP-reduction}
\label{sec:confluence}

In this section we introduce the question of confluence of singular DP-reduction. In Subsection \ref{sec:questconfl} we define ``confluence'' and ``confluence modulo isomorphism'', and discuss basic examples. In Subsection \ref{sec:confluencesatmu} we obtain our first major result, namely confluence for $\Smusat$ (Theorem \ref{thm:confsat}).

\subsection{The question of confluence}
\label{sec:questconfl}

\begin{definition}\label{def:conflsdp}
  Let \bmm{\cflmusat} be the set of $F \in \Musat$ where singular DP-reduction is confluent, and let \bmm{\cflimusat} be the set of $F \in \Musat$ where singular DP-reduction is confluent modulo isomorphism:
  \begin{eqnarray*}
    \cflmusat & := & \set{F \in \Musat \mb \abs{\sdp(F)} = 1}\\
    \cflimusat & := & \set{F \in \Musat \mb \fa\, F',F'' \in \sdp(F) : F' \cong F''}.
  \end{eqnarray*}
\end{definition}
\begin{example}\label{exp:conflsdp} Examples illustrating $\cflmusat \subset \cflimusat \subset \Musat$:
  \begin{enumerate}
  \item In \cite{DDK98} it is shown that every $F \in \Musati{\delta=1}$ contains a 1-singular variable (see \cite{Ku99dKo,KullmannZhao2011Bounds} for further generalisations). Thus by Corollary \ref{cor:sdppresmu} we get that singular DP-reduction on $\Musati{\delta=1}$ must end in $\set{\set{\bot}}$, and we have $\Musatnsi{\delta=1} = \set{\set{\bot}}$. It follows $\Musati{\delta=1} \sse \cflmusat$.
  \item\label{exp:conflsdp1} We now show $\Musati{\delta=2} \not\sse \cflmusat$. Let $F \in \Musati{\delta=2}$ be obtained from $\Dt{2}$ (recall Example \ref{exp:F2F3}) by ``inverse singular DP-reduction'', adding a new singular variable $v$ and replacing the two clause $\set{v_1,v_2}, \set{\ol{v_2},v_1} \in \Dt{2}$ by the three clauses $\set{v,v_1}, \set{\ol{v},v_2}, \set{\ol{v},\ol{v_2}}$, obtaining $F$ (the other two clauses in $F$ are $\set{\ol{v_1},v_2}, \set{\ol{v_1},\ol{v_2}}$):
    \begin{displaymath}
      F = \setb{ \set{v,v_1}, \set{\ol{v},v_2}, \set{\ol{v},\ol{v_2}}, \set{\ol{v_1},v_2}, \set{\ol{v_1},\ol{v_2}} }.
    \end{displaymath}
Singular DP-reduction on $v$ yields $\Dt{2}$ (and thus by Lemma \ref{lem:singDpMU} we get indeed $F \in \Musati{\delta=2}$). The second singular variable of $F$ is $v_1$, and sDP-reduction on $v_1$ yields $F' := \set{\set{v,v_2},\set{v,\ol{v_2}},\set{\ol{v},v_2}, \set{\ol{v},\ol{v_2}}}$, where $F' \not= \Dt{2}$. Note however that we have $F' \cong \Dt{2}$ (since $F'$ consists of all binary clauses over the variables $v,v_2$), and in Theorem \ref{thm:confmodisomu2} we will indeed see that we have $\Musati{\delta=2} \sse \cflimusat$.
  \item\label{exp:conflsdp2} We show $\Musati{\delta=3} \not\sse \cflimusat$ by constructing $F \in \Musati{\delta=3}$ with $\sdp(F) = \set{F_1,F_2}$ where $F_1 \not\cong F_2$. Let $G_1 := \Dt{2}$, and let $G_2$ be the variable-disjoint copy of $G_1$ obtained by replacing variables $v_1,v_2$ with $v_1',v_2'$. Let $w$ be a new variable, and obtain $F_1$ by ``full gluing'' of $G_1, G_2$ on $w$, that is, add literal $w$ to all clauses of $G_1$, add literal $\ol{w}$ to all clauses of $G_2$, and let $F_1$ be the union of these two clause-sets:
    \begin{multline*}
      F_1 = \setb{ \set{w,v_1,v_2},\set{w,v_1,\ol{v_2}}, \set{w,\ol{v_1},v_2}, \set{w,\ol{v_1},\ol{v_2}},\\
        \set{\ol{w},v_1',v_2'}, \set{\ol{w},v_1',\ol{v_2'}},\set{\ol{w},\ol{v_1'},v_2'},\set{\ol{w},\ol{v_1'},\ol{v_2'}}
      }.
    \end{multline*}
    We have $F_1 \in \Uclashnsi{\delta=3}$. We obtain $F$ from $F_1$ by inverse singular DP-reduction, adding a new (singular) variable $v$, and replacing the two clauses $\set{w,v_1,v_2}, \set{w,v_1,\ol{v_2}}$ by the three clauses $\set{v,w,v_1},\set{\ol{v},v_2},\set{\ol{v},w,\ol{v_2}}$:
    \begin{multline*}
      F = \setb{ \set{v,w,v_1},\set{\ol{v},v_2},\set{\ol{v},w,\ol{v_2}}, \ \set{w,\ol{v_1},v_2}, \set{w,\ol{v_1},\ol{v_2}},\\
        \set{\ol{w},v_1',v_2'}, \set{\ol{w},v_1',\ol{v_2'}},\set{\ol{w},\ol{v_1'},v_2'},\set{\ol{w},\ol{v_1'},\ol{v_2'}}
      }.
    \end{multline*}
Singular DP-reduction on $v$ yields $F_1$, and thus $F \in \Musati{\delta=3}$. The second singular variable of $F$ is $v_1$, and sDP-reduction on $v_1$ yields a clause-set $F_2$ containing one binary clause (since we left out $w$ in the replacement-clause $\set{\ol{v},v_2}$). Since all clauses in $F_1$ have length $3$, we see $F_2 \not\cong F_1$.
  \end{enumerate}
\end{example}

\subsection{Confluence on saturated $\Musat$}
\label{sec:confluencesatmu}

\begin{definition}\label{def:refinement}
  For clause-sets $F,G$ we write \bmm{F \sselr G} if for all $C \in F$ there is $D \in G$ with $C \sse D$.
\end{definition}
If $F \sselr G$, then we say that ``$F$ is a subset of $G$ mod(ulo) supersets''. $\sselr$ is a quasi-order on arbitrary clause-sets and a partial order on subsumption-free clause-sets, and thus $\sselr$ is a partial order on $\Musat$. The minimal element of $\sselr$ on $\Cls$ is $\top$, the minimal element on $\Musat$ is $\set{\bot}$. Now we show that ``nonsingular saturated patterns'' are not destroyed by singular DP-reduction:
\begin{lemma}\label{lem:uniquenf}
  Consider $F_0, F, F' \in \Musat$ with $F \tsdps F'$.
  \begin{enumerate}
  \item\label{lem:uniquenf1} If $F_0$ is nonsingular, then $F_0 \sselr F \Ra F_0 \sselr F'$.
  \item\label{lem:uniquenf2} If $F_0, F, F' \in \Smusat$, then $F_0 \sselr F' \Ra F_0 \sselr F$.
  \end{enumerate}
\end{lemma}
\begin{proof} W.l.o.g.\ we can assume for both parts that $F' = \dpi{v}(F)$ for a singular variable $v$ of $F$. Part \ref{lem:uniquenf1} follows from the facts that $v \notin \var(F_0)$ due to the nonsingularity of $F_0$, and that due to the minimal unsatisfiability of $F$ no clause gets lost by an application of singular DP-reduction. For Part \ref{lem:uniquenf2} assume $\ldeg_F(v) = 1$. Then the assertion follows from the fact, that due to the saturatedness of $F$ we have for the clause $C \in F$ with $v \in C$ and for every clause $D \in F$ with $\ol{v} \in D$ that $C \sm \set{v} \sse D \sm \set{\ol{v}}$. \qed
\end{proof}
\begin{example}\label{exp:uniquenf}
  Illustrating the conditions of Lemma \ref{lem:uniquenf}:
  \begin{enumerate}
  \item An example showing that in Part \ref{lem:uniquenf1} nonsingularity of $F_0$ is needed, is given trivially by $F = F_0 = \set{\set{v},\set{\ol{v}}}$.
  \item While an example for Part \ref{lem:uniquenf2} with $F \in \Musat \sm \Smusat$ and $F_0 \not\sselr F$ is given by $F_0 = F' = \Dt{3}$ (recall Example \ref{exp:F2F3}) and
    \begin{displaymath}
      F = \set{\set{\ol{v_1},\ol{v_2},\ol{v_3}}, \set{v_1,v_2,v},\set{\ol{v},v_3},\set{\ol{v_1},v_2},\set{\ol{v_2},v_3},\set{\ol{v_3},v_1}}.
    \end{displaymath}
  \end{enumerate}
\end{example}

\begin{theorem}\label{thm:confsat}
  $\Smusat \subset \cflmusat$.
\end{theorem}
\begin{proof}
  Consider $F \in \Smusat$ and two nonsingular $F', F'' \in \Smusat$ with $F \tsdps F'$ and $F \tsdps F''$. From $F'\sselr F'$ and $F \tsdps F'$ by Lemma \ref{lem:uniquenf}, Part \ref{lem:uniquenf2} we get $F' \sselr F$, and then by Part \ref{lem:uniquenf1} we get $F' \sselr F''$; in the same way we obtain $F'' \sselr F'$ and thus $F' = F''$. \qed
\end{proof}

\section{Permutations of sequences of DP-reductions}
\label{sec:permDPred}

This section contains central technical results on (iterated) singular DP-reduction. The basic observations are collected in Subsection \ref{sec:propsingdpred}, studying how literal degrees change under sDP-reductions. It follows an interlude on iterated general DP-reduction in Subsection \ref{sec:itDPr}, stating ``commutativity modulo subsumption'' and deriving the basic fact in Corollary \ref{cor:dpeqpres}, that in case a sequence of DP-reductions as well as some permutation both yield minimally unsatisfiable clause-sets, then actually these MU's are the same. In Subsection \ref{sec:itsDP} then conclusions for singular DP-reductions are drawn, obtaining various conditions under which sDP-reductions can be permuted without changing the final result. A good overview on all possible sDP-reductions is obtained in Subsection \ref{sec:without1sing} in case no 1-singular variables are present. In Subsection \ref{sec:singindex} we introduce the ``singularity index'', the minimal length of a maximal sDP-reduction sequence. Our second major result is Theorem \ref{thm:singind}, showing that in fact all maximal sDP-reduction-sequences have the same length.

\subsection{Monitoring literal degrees under singular DP-reductions}
\label{sec:propsingdpred}

First we analyse the changes for literal-degrees after one step of sDP-reduction.
\begin{lemma}\label{lem:sdpld}
  Consider $F \in \Musat$ and an $m$-singular variable $v$ ($m \in \NN$). Let $C$ be the main clause, and let $D_1, \dots, D_m$ be the side clauses; and let $F' := \dpi{v}(F)$. Consider a literal $x \in \Lit$; the task is to compare $\ldeg_F(x)$ and $\ldeg_{F'}(x)$.
  \begin{enumerate}
  \item\label{lem:sdpld1} If $\var(x) \not= v$ and $x \notin C$, then $\ldeg_{F'}(x) = \ldeg_F(x)$.
  \item\label{lem:sdpld2} If $\var(x) = v$, then $\ldeg_F(x) + \ldeg_F(\ol{x}) = m+1$, while $\ldeg_{F'}(x) = \ldeg_{F'}(\ol{x}) = 0$.
  \item[] For the remaining items we assume $\var(x) \not= v$ and $x \in C$.
  \item[] Let $p := \abs{\set{i \in \tb 1m : x \notin D_i}} \in \tb 0m$.
  \item\label{lem:sdpld3} $\ldeg_{F'}(x) = \ldeg_F(x) - 1 + p$.
  \item\label{lem:sdpld3a} $\max(m, \ldeg_F(x) - 1) \le \ldeg_{F'}(x) \le \ldeg_F(x) - 1 + m$.
  \item\label{lem:sdpld7} If $m=1$, then $\ldeg_{F'}(x) \le \ldeg_F(x)$.
  \item\label{lem:sdpld6} We have $\ldeg_{F'}(x) > \ldeg_F(x)$ iff $p \ge 2$.
  \item\label{lem:sdpld4} The following conditions are equivalent:
    \begin{enumerate}
    \item $\ldeg_{F'}(x) = \ldeg_F(x) - 1$.
    \item $\ldeg_{F'}(x) < \ldeg_F(x)$.
    \item $x \in C \cap D_1 \cap \dots \cap D_m$.
    \item $\var(x) \in \var(C) \cap \var(D_1) \cap \dots \cap \var(D_m)$.
    \end{enumerate}
  \item\label{lem:sdpld5} If $\ldeg_{F'}(x) < \ldeg_F(x)$, then $\ldeg_{F'}(x) = \ldeg_F(x) - 1 \ge m$.
  \end{enumerate}
\end{lemma}
\begin{proof} Parts \ref{lem:sdpld1} - \ref{lem:sdpld3} follow by definition, Parts \ref{lem:sdpld3a}, \ref{lem:sdpld6} follows by Part \ref{lem:sdpld3}, Part \ref{lem:sdpld7} follows by Part \ref{lem:sdpld3a}. Part \ref{lem:sdpld4} follows by Parts \ref{lem:sdpld1} - \ref{lem:sdpld3a} and the observation, that if $x \in C$, then $\ol{x} \notin D_1 \cup \dots \cup D_m$ (due to $F \in \Musat$). Part \ref{lem:sdpld5} follows by Part \ref{lem:sdpld4}. \qed
\end{proof}

By Lemma \ref{lem:sdpld}, Parts \ref{lem:sdpld5} and \ref{lem:sdpld7} we get that singular variables can only be created for 1-singular DP-reduction, while singular variables can only be destroyed for non-1-singular DP-reductions; the details are as follows:
\begin{corollary}\label{cor:propdps3}
  Consider $F \in \Musat$ and an $m$-singular variable $v$ for $F$ ($m \in \NN$), and let $F' := \dpi{v}(F)$.
  \begin{enumerate}
  \item\label{cor:propdps31}
    \begin{enumerate}
    \item\label{cor:propdps311} If $m \ge 2$, then $\varsing(F') \sse \varsing(F)$ with $\varosing(F') \sse \varosing(F)$.
    \item\label{cor:propdps312} If $m = 1$ then $\varosing(F') \sm \varosing(F) \sse \varnosing(F)$.
    \end{enumerate}
  \item\label{cor:propdps32}
    \begin{enumerate}
    \item\label{cor:propdps321} If $m=1$, then $\varsing(F) \sm \set{v} \sse \varsing(F')$ with $\varosing(F) \sm \set{v} \sse \varosing(F')$.
    \item\label{cor:propdps322} If $m \ge 2$ then $\varsing(F) \sm \varsing(F') \sse \varnosing(F)$.
    \end{enumerate}
  \end{enumerate}
\end{corollary}

By Lemma \ref{lem:sdpld}, Part \ref{lem:sdpld6} together with Lemma \ref{lem:singDpSMU} we get:
\begin{corollary}\label{cor:propdps3b}
  Consider $F \in \Smusat$ and a singular variable $v$; let $F' := \dpi{v}(F)$.
  \begin{enumerate}
  \item\label{cor:propdps3b1} For all literals $x$ holds $\ldeg_{F'}(x) \le \ldeg_F(x)$.
  \item\label{cor:propdps3b2} Thus if $w \not= v$ is a singular variable for $F$, then $w$ is also singular for $F'$.
  \end{enumerate}
\end{corollary}

\subsection{Iterated DP-reduction}
\label{sec:itDPr}

\begin{definition}\label{def:itDP}
  Consider $F \in \Cls$ and a sequence $v_1, \dots, v_n$ of variables for $n \in \NNZ$. Then\vspace{-2ex}
  \begin{displaymath}
    \bmm{\dpi{v_1,\dots,v_n}(F)} :=
    \begin{cases}
      F & \text{if } n = 0\\
      \dpi{v_n}(\dpi{v_1,\dots,v_{n-1}}(F)) & \text{if } n > 0
    \end{cases}.
  \end{displaymath}
\end{definition}
Thus in ``$\dpi{v_1,\dots,v_n}$'' DP-reduction is performed in order $v_1, \dots, v_n$. We have $\var(\dpi{v_1,\dots,v_n}(F)) \sse \var(F) \sm \set{v_1,\dots,v_n}$. In \cite{KuLu97} (Lemma 7.4, page 33) as well as in \cite{KuLu98} (Lemma 7.6, page 27) the following fundamental result on iterated DP-reduction is shown:  If performing subsumption-elimination at the end, then iterated DP-reduction does not depend on the order of the variables, while additionally performing subsumption-elimination inbetween has no influence. More precisely:
\begin{lemma}\label{lem:DPcomm}
  Let $\rsub: \Cls \ra \Cls$ be subsumption-elimination, that is, $\rsub(F)$ is the set of $C \in F$ which are minimal in $F$ w.r.t.\ the subset-relation. And for $n \in \NNZ$ let $S_n$ be the set of permutations of $\tb 1n$. Then we have the following operator-equalities for all variable-sequences $v_1, \dots, v_n \in \Va$ ($n \in \NNZ$):
  \begin{enumerate}
  \item\label{lem:DPcomm1} $\rsub \circ \dpi{v_1,\dots,v_n} = \rsub \circ \dpi{v_1,\dots,v_n} \circ \rsub$.
  \item\label{lem:DPcomm2} For all $\pi \in S_n$ we have $\rsub \circ \dpi{v_1,\dots,v_n} = \rsub \circ \dpi{v_{\pi(1)},\dots,v_{\pi(n)}}$.
  \end{enumerate}
\end{lemma}

\begin{definition}\label{def:DPeqpres}
  Consider $F \in \Cls$ and $v_1,\dots,v_n \in \Va$ ($n \in \NNZ$). Then a permutation $\pi \in S_n$ is called \textbf{equality-preserving} for $F$ and $v_1,\dots,v_n$ (for short: ``eq-preserving''), if we have $\dpi{v_1,\dots,v_n}(F) = \dpi{\pi(v_1),\dots,\pi(v_n)}(F)$. The set of all eq-preserving $\pi \in S_n$ is denoted by $\bmm{\eqp(F, (v_1,\dots,v_n))} \sse S_n$.
\end{definition}
Note that if $\var(F) \sse \set{v_1,\dots,v_n}$, then $\eqp(F,(v_1,\dots,v_n)) = S_n$. Since minimally unsatisfiable clause-sets do not contain subsumptions, we obtain:
\begin{corollary}\label{cor:dpeqpres}
  Consider $F \in \Cls$ and variables $v_1, \dots, v_n$ ($n \in \NNZ$) such that $\dpi{v_1,\dots,v_n}(F) \in \Musat$. Then we have for $\pi \in S_n$:
  \begin{displaymath}
    \pi \in \eqp(F,(v_1,\dots,v_n)) \Lra \dpi{v_{\pi(1)},\dots,v_{\pi(n)}}(F) \in \Musat.
  \end{displaymath}
\end{corollary}
Since hitting clause-sets do not contain subsumptions, by Lemma \ref{lem:hitDPg} we obtain:
\begin{corollary}\label{cor:dpeqhit}
  For clause-sets $F \in \Clash$ and variables $v_1, \dots, v_n$ ($n \in \NNZ$) we have $\eqp(F,(v_1,\dots,v_n)) = S_n$.
\end{corollary}

\subsection{Iterated sDP-reduction via singular tuples}
\label{sec:itsDP}

Generalising Definition \ref{def:singV1} we consider ``singular tuples'':
\begin{definition}\label{def:singseq}
  Consider $F \in \Musat$. A tuple $\vec{v} = (v_1, \dots, v_n)$ of variables ($n \in \NNZ$) is called \textbf{singular} for $F$ if for all $i \in \tb 1n$ we have that $v_i$ is singular for $\dpi{v_1,\dots,v_{i-1}}(F)$. Note that for a singular tuple $(v_1,\dots,v_n)$ all variables must be different. We call variable $v_i$ ($i \in \tb 1n$) \textbf{$m$-singular} ($m \in \NN$) for $\vec{v}$ and $F$, if $v_i$ is $m$-singular for $\dpi{v_1,\dots,v_{i-1}}(F)$. And the \textbf{singularity-degree tuple} of $\vec{v}$ w.r.t.\ $F$ is the tuple $(m_1,\dots,m_n)$ of natural numbers such that $v_i$ is $m_i$-singular for $\vec{v}$ and $F$.
\end{definition}
\begin{example}\label{exp:singvar_cont}
  Consider $F := \set{\set{a},\set{\ol{a},b},\set{\ol{a},\ol{b}}}$ (recall Example \ref{exp:singvar}). There are $5$ singular tuples for $F$, namely $(),(a),(b),(a,b),(b,a)$. Considering $\vec{v} := (a,b)$, variable $a$ is $2$-singular for $\vec{v}$ and $F$, and $b$ is $1$-singular for $\vec{v}$ and $F$, and thus its singularity-degree sequence is $(2,1)$, while considering $\vec{v}' := (b,a)$, both $a$ and $b$ are  $1$-singular for $\vec{v}'$ and $F$, and thus the singularity-degree sequence is $(1,1)$.
\end{example}
For the understanding of sDP-reduction of $F \in \Musat$, understanding the set of singular tuples for $F$ is an important task. Two basic properties are:
\begin{enumerate}
\item $F$ has only the empty singular tuple iff $F$ is nonsingular.
\item If $(v_1,\dots,v_n)$ is a singular tuple for $F$, then for all $i \in \tb 0n$ the tuple $(v_1,\dots,v_i)$ is also singular for $F$.
\end{enumerate}

\begin{definition}\label{def:singpres}
  Consider $F \in \Musat$ and a singular tuple $(v_1,\dots,v_n)$ for $F$. A permutation $\pi \in S_n$ is called \textbf{singularity-preserving} for $F$ and $(v_1,\dots,v_n)$ (for short: ``s-preserving''), if also $(v_{\pi(1)},\dots,v_{\pi(n)})$ is singular for $F$. The set of all s-preserving $\pi \in S_n$ is denoted by $\bmm{\sgp(F, (v_1,\dots,v_n))} \sse S_n$.
\end{definition}
By Corollary \ref{cor:dpeqpres} we obtain the fundamental lemma, showing that singularity-preservation implies equality-preservation:
\begin{lemma}\label{lem:speqp}
  For $F \in \Musat$ and a singular tuple $\vec{v}$ we have $\sgp(F, \vec{v}) \sse \eqp(F,\vec{v})$.
\end{lemma}
Thus singular tuples with the same variables yield the same reduction-result:
\begin{corollary}\label{cor:singeqvar}
  Consider two singular tuples $(v_1,\dots,v_n), (v_1',\dots,v_n')$ for $F \in \Musat$. If $\set{v_1,\dots,v_n} = \set{v_1',\dots,v_n'}$, then $\dpi{v_1,\dots,v_n}(F) = \dpi{v_1',\dots,v_n'}(F)$.
\end{corollary}

Preparing our results on singularity-preserving permutations, we consider first ``homogeneous'' singular pairs in the following two (easy) lemmas.
\begin{lemma}\label{lem:22sing}
  Consider $F \in \Musat$ and two different non-1-singular variables $v, w$ for $F$. Let $C$ be the main clause for $v$, and let $D$ be the main clause for $w$. There are precisely two cases now:
  \begin{enumerate}
  \item\label{lem:22sing1} If $C = D$, then $w \notin \varsing(\dpi{v}(F))$ and $v \notin \varsing(\dpi{w}(F))$.
  \item\label{lem:22sing2} If $C \not= D$, then $w \in \varnosing(\dpi{v}(F))$ and $v \in \varnosing(\dpi{w}(F))$.
  \end{enumerate}
\end{lemma}
\begin{proof}
  Part \ref{lem:22sing1} follows by Lemma \ref{lem:sdpld}, Part \ref{lem:sdpld3a}, and Part \ref{lem:22sing2} follows by Part \ref{lem:sdpld5} of that lemma (to see that the complements occur at least twice after the DP-reductions). \qed
\end{proof}

\begin{example}\label{exp:twonosingv}
  We illustrate the two cases of Lemma \ref{lem:22sing}:
  \begin{enumerate}
  \item Let $F := \set{\set{v,w},\set{\ol{v},a},\set{\ol{v},\ol{a}},\set{\ol{w},b},\set{\ol{w},\ol{b}}} \in \Musati{\delta=1}$. Then $C = D = \set{v,w}$, and $w$ is not singular in $\dpi{v}(F) = \set{\set{w,a},\set{w,\ol{a}},\set{\ol{w},b},\set{\ol{w},\ol{b}}}$, and $v$ is not singular in $\dpi{w}(F) = \set{\set{\ol{v},a},\set{\ol{v},\ol{a}},\set{v,b},\set{v,\ol{b}}}$.
  \item Let $F := \set{\set{v,a},\set{w,\ol{a}},\set{\ol{v},a,b},\set{\ol{v},a,\ol{b}},\set{\ol{w},\ol{a},b},\set{\ol{w},\ol{a},\ol{b}}} \in \Smusati{\delta=2}$. Then $C = \set{v,a} \not= D = \set{w,\ol{a}}$, and now $w$ is singular in $\dpi{v}(F) = \set{\set{w,\ol{a}},\set{a,b},\set{a,\ol{b}},\set{\ol{w},\ol{a},b},\set{\ol{w},\ol{a},\ol{b}}}$, and $v$ is singular in $\dpi{w}(F) = \set{\set{v,a},\set{\ol{v},a,b},\set{\ol{v},a,\ol{b}},\set{\ol{a},b},\set{\ol{a},\ol{b}}}$.
  \end{enumerate}
\end{example}

\begin{lemma}\label{lem:11sing}
  Consider $F \in \Musat$ and a singular tuple $(v,w)$ for $F$ with singularity-degree tuple $(1,1)$. Let $C, D \in F$ be the two occurrences of $v$.
  \begin{enumerate}
  \item\label{lem:11sing1} Assume $w$ is not 1-singular in $F$:
    \begin{enumerate}
    \item Then $w$ is $2$-singular in $F$. Let $E_0 \in F$ be the main-clause of $w$, and let $E_1, E_2 \in F$ be the two side-clauses.
    \item We have $\set{E_1, E_2} = \set{C,D}$.
    \item So $v$ is $1$-singular in $\dpi{w}(F)$.
    \item Thus $(w,v)$ is a singular tuple with singularity-degree tuple $(2,1)$.
    \end{enumerate}
  \item\label{lem:11sing2} Otherwise $w$ is 1-singular in $F$.
    \begin{enumerate}
    \item $v$ is $1$-singular in $\dpi{w}(F)$.
    \item Thus $(w,v)$ is a singular tuple with singularity-degree tuple $(1,1)$.
    \item Let $E_1, E_2$ be the two occurrences of $w$ in $F$: $\abs{\set{C,D} \cap \set{E_1,E_2}} \le 1$.
    \end{enumerate}
  \end{enumerate}
\end{lemma}
\begin{proof}
  For Part \ref{lem:11sing1} we use Lemma \ref{lem:sdpld}, Part \ref{lem:sdpld4}, and we see that $w$ is $2$-singular for $F$ (sDP-reduction can only reduce literal-degrees by one), and the complement of the singular literal of $w$ must occur in all occurrences of variable $v$; we also see that DP-reduction on $w$ does not change the degree of $v$. For Part \ref{lem:11sing2} we use Corollary \ref{cor:propdps3}, Part \ref{cor:propdps321}, together we the fact that the occurrences of two 1-singular variables can not completely coincide, since then we had more than one clash between the main clause and the side clause (see Lemma \ref{lem:singDpMU}, Part \ref{lem:singDpMU3a}). \qed
\end{proof}

\begin{example}\label{exp:twoosingv}
  We illustrate the two cases of Lemma \ref{lem:11sing}:
  \begin{enumerate}
  \item Let $F := \set{\set{v,w},\set{\ol{v},w},\set{\ol{w}}} \in \Smusati{\delta=1}$, with $\dpi{v}(F) = \set{\set{w},\set{\ol{w}}}$. Then $\set{C,D} = \set{\set{v,w},\set{\ol{v},w}}$, and $E_0 = \set{\ol{w}}$ and $\set{E_1,E_2} = \set{C,D}$, where $\dpi{w}(F) = \set{\set{v},\set{\ol{v}}}$.
  \item We give examples for both cases of $\abs{\set{C,D} \cap \set{E_1,E_2}} \in \set{0,1}$:
    \begin{enumerate}
    \item Let $F := \set{\set{v,a},\set{\ol{v},a},\set{w,\ol{a}},\set{\ol{w},\ol{a}}} \in \Smusati{\delta=1}$. Then $\dpi{v}(F) = \set{\set{a},\set{w,\ol{a}},\set{\ol{w},\ol{a}}}$ and $\dpi{w}(F) = \set{\set{v,a},\set{\ol{v},a},\set{\ol{a}}}$, where $\set{C,D} = \set{\set{v,a},\set{\ol{v},a}}$ and $\set{E_1,E_2} = \set{\set{w,\ol{a}},\set{\ol{w},\ol{a}}}$.
    \item Let $F := \set{\set{v},\set{\ol{v},w},\set{\ol{w}}} \in \Musati{\delta=1}$. Then $\dpi{v}(F) = \set{\set{w},\set{\ol{w}}}$ and $\dpi{w}(F) = \set{\set{v},\set{\ol{v}}}$, where $\set{C,D} = \set{\set{v},\set{\ol{v},w}}$ and $\set{E_1,E_2} = \set{\set{\ol{v},w},\set{\ol{w}}}$.
    \end{enumerate}
  \end{enumerate}
\end{example}

Now we are ready to show the central ``exchange theorem'', characterising s-preserving neighbour-exchanges (recall that every permutation is a composition of neighbour-exchanges): The gist of Theorem \ref{thm:nbesp} is that in most cases neighbours in a singular tuple can be exchanged safely (i.e., s-preserving), except of the cases where a 1-singular DP-reduction is followed by a non-1-singular DP-reduction (Case \ref{thm:nbesp32}).
\begin{theorem}\label{thm:nbesp}
  Consider $F \in \Musat$ and a singular tuple $\vec{v} = (v_1, \dots, v_n)$ with $n \ge 2$, and let $(m_1,\dots,m_n)$ be the singularity-degree tuple of $\vec{v}$ w.r.t.\ $F$. Consider $i \in \tb 1{n-1}$, and let $\pi \in S_n$ be the neighbour-exchange $i \lra i+1$ (i.e., $\pi(j) = j$ for $j \in \tb 1n \sm \set{i,i+1}$, while $\pi(i) = i+1$ and $\pi(i+1) = i$). Let $(m_1',\dots,m_n')$ be the singularity-degree tuple of $\vec{v}'$ w.r.t.\ $F$, where $\vec{v}' := (v_{\pi(1)},\dots,v_{\pi(n)})$, in case of $\pi \in \sgp(F,\vec{v})$. The task is to characterise when $\pi \in \sgp(F,\vec{v})$ holds; we also need to be able to apply such s-preserving neighbour-exchanges consecutively, by controlling the changes in the singularity-degrees.
  \begin{enumerate}
  \item\label{thm:nbesp1} If $\pi \in \sgp(F,\vec{v})$, then for $j \in \tb 1n \sm \set{i,i+1}$ we have $m_j' = m_j$.
  \item\label{thm:nbesp2} Assume $m_i \ge 2$.
    \begin{enumerate}
    \item\label{thm:nbesp21} $\pi \in \sgp(F,\vec{v})$.
    \item\label{thm:nbesp22newa} $m_i' \le m_{i+1}+1$.
    \item\label{thm:nbesp22newb} $m_{i+1}' \ge m_i-1$.
    \item\label{thm:nbesp23} If $m_{i+1} = 1$, then $m_i' = 1$.
    \item\label{thm:nbesp23newa} If $m_{i+1} \ge 2$, then $m_{i+1}' \ge 2$.
    \end{enumerate}
  \item\label{thm:nbesp3} Assume $m_i = 1$.
    \begin{enumerate}
    \item\label{thm:nbesp31} Assume $m_{i+1} = 1$.
      \begin{enumerate}
      \item\label{thm:nbesp311} $\pi \in \sgp(F,\vec{v})$.
      \item\label{thm:nbesp312} $m_{i+1}' = 1$ and $m_i' \in \set{1,2}$.
      \end{enumerate}
    \item\label{thm:nbesp32} Assume $m_{i+1} \ge 2$.
      \begin{enumerate}
      \item\label{thm:nbesp321} $\pi \in \sgp(F,\vec{v})$ if and only if $v_{i+1}$ is singular in $\dpi{v_1,\dots,v_{i-1}}(F)$.
      \item\label{thm:nbesp322} If $\pi \in \sgp(F,\vec{v})$, then $m_i' \ge 2$.
      \end{enumerate}
    \end{enumerate}
  \end{enumerate}
\end{theorem}
\begin{proof}
Part \ref{thm:nbesp1} follows by Lemma \ref{lem:speqp}. For the remainder let $F_0 := F$, and $F_i := \dpi{v_i}(F_{i-1})$ for $i \in \tb 1n$.

Now consider Part \ref{thm:nbesp2}; so we assume $m_i \ge 2$ here. For Part \ref{thm:nbesp21} we need to show that $v_{i+1}$ is singular for $F_i$ and $v_i$ is singular for $\dpi{v_{i+1}}(F_i)$: The former follows by Corollary \ref{cor:propdps3}, Part \ref{cor:propdps311}, while the latter follows by Part \ref{cor:propdps321} of that Corollary (if $v_{i+1}$ is 1-singular for $F_i$) and by both parts of Lemma \ref{lem:22sing} (if $v_{i+1}$ is non-1-singular for $F_i$; the main clauses for $v_i, v_{i+1}$ in $F_i$ can not be the same).

Part \ref{thm:nbesp22newa}, \ref{thm:nbesp22newb} follow by Part \ref{lem:sdpld4} of Lemma \ref{lem:sdpld}, while Part \ref{thm:nbesp23} follows by Part \ref{lem:sdpld5} of that Lemma. Now consider Part \ref{thm:nbesp23newa}, and so we assume $m_{i+1} \ge 2$. If $m_i' \ge 2$ then $m_{i+1}' \ge 2$ follows from Part \ref{cor:propdps311} of Corollary \ref{cor:propdps3}; it remains the case $m_i' = 1$. Let $x$ be the singular literal of $v_i$ in $F_i$, and let $y$ be the singular literal of $v_{i+1}$ in $F_{i+1}$. Since sDP-reduction by $v_i$ in $F_i$ increased the number of occurrences of $\ol{y}$, for the main clause $C$ of $v_i$ in $F_i$ (thus $x \in C$) we must have $\ol{y} \in C$. Let $D$ be the main clause of $v_{i+1}$ in $F_i$, that is, $y \in D$ (note that $C, D$ are the only occurrences of variable $v_{i+1}$ in $F_i$). If $m_{i+1}' = 1$ would be the case, then we would have $\ol{x} \in C, D$ contradicting $x \in C$.

Finally consider Part \ref{thm:nbesp3}, assuming $m_i = 1$. Part \ref{thm:nbesp31} follows with Lemma \ref{lem:11sing}. For Part \ref{thm:nbesp32} assume $m_{i+1} \ge 2$. For Part \ref{thm:nbesp321} the direction from left to right follows by definition, while the direction from right to left follows by Part \ref{cor:propdps322} of Lemma \ref{cor:propdps3}. And Part \ref{thm:nbesp322} by Part \ref{lem:sdpld7} of Lemma \ref{lem:sdpld}. \qed
\end{proof}
We remark that for Part \ref{thm:nbesp23newa} of Theorem \ref{thm:nbesp}, in the conference version we also asserted that $m_i' \ge 2$ would be the case (Lemma 26, Part 2, in \cite{KullmannZhao2012ConfluenceC}), which is false as shown in Example \ref{exp:turn12}.

\begin{corollary}\label{cor:suffcondswap}
  Consider $F \in \Musat$ and a singular tuple $\vec{v} = (v_1,\dots,v_n)$ ($n \ge 2$) with $1 \le i < n$. Then a sufficient condition for the neighbour exchange $i \lra i+1$ to be s-preserving is:\vspace{-1ex}
  \begin{center}
    $v_i$ is non-1-singular for $\vec{v}$, or $v_{i+1}$ is 1-singular for $\vec{v}$,\\
    or $v_{i+1}$ is singular for $\dpi{v_1,\dots,v_{i-1}}(F)$.
  \end{center}
\end{corollary}

\subsubsection{Examples}
\label{sec:exchangelemmaExamples}

We now give various examples showing that the bounds from Theorem \ref{thm:nbesp} are sharp in general. First we show that a swap of two non-1-singular variables can create a 1-singular variables.
\begin{example}\label{exp:turn12}
  Consider $k \in \NN$. The following $F \in \Musat$ and $v, w \in \var(F)$ have the properties that $(v,w)$ is a singular tuple with singularity-degree tuple $(k,k)$ while $(w,v)$ is a singular tuple with singularity-degree tuple $(1,k)$.
  \begin{enumerate}
  \item Let $F := \setb{\set{v,w},\set{\ol{v},x_1},\dots,\set{\ol{v},x_k},\set{\ol{w}},\set{\ol{x_1},\dots,\ol{x_k}}} \in \Musati{\delta=1}$.
  \item $v$ is $k$-singular for $F$, while $w$ is 1-singular for $F$.
  \item We have $\varsing(F) = \var(F) = \set{v,w,x_1,\dots,x_k}$ and $\varnosing(F) = \set{v}$.
  \item Let $F' := \dpi{v}(F) = \set{\set{w,x_1},\dots,\set{w,x_k},\set{\ol{w}},\set{\ol{x_1},\dots,\ol{x_k}}}$.
  \item Now $w$ is $k$-singular for $F'$, and thus the associated singularity-degree tuple for $(v,w)$ and $F$ is $(k,k)$.
  \item While the singular tuple $(w,v)$ has singularity-degree tuple $(1,k)$.
  \end{enumerate}
\end{example}
Next we give examples showing that the bounds from Part \ref{thm:nbesp2} of Theorem \ref{thm:nbesp} are sharp in general.
\begin{example}\label{exp:nbesp2}
   All examples (again) are in $\Musati{\delta=1}$.
  \begin{enumerate}
  \item First we consider Part \ref{thm:nbesp22newa}, showing that the two extreme cases $m_i' = 1$ and $m_i' = m_{i+1}+1$ are possible.
    \begin{enumerate}
    \item Example \ref{exp:turn12} yields $m_i = m_{i+1} = k \ge 2$ and $m_i' = 1$, $m_{i+1}' = k$.
  \item That is, the original pair $(v_i,v_{i+1})$ has singularity-degree tuple $(k,k)$, while after swap we have $(1,k)$. In the sequel we will describe the examples in this manner.
  \item For $k \in \NN$ let $F_1 := \set{\set{v,\ol{w}},\set{\ol{v},\ol{w},x_1},\dots,\set{\ol{v},\ol{w},x_k},\set{\ol{x_1},\dots,\ol{x_k}},\set{w}}$. Then for $(v,w)$ we have $(k,k)$, while for $(w,v)$ we have $(k+1,k)$.
    \end{enumerate}
  \item Now we consider Part \ref{thm:nbesp22newb}, showing that $m_{i+1}' = m_i-1+p$ for all $p \in \NNZ$ is possible.
    \begin{enumerate}
    \item For $p=0$ we just re-use $F_1$, but in the other direction, from $(w,v)$ with $(k+1,k)$ to $(v,w)$ with $(k,k)$.
    \item Let $F_2 := \set{\set{v},\set{\ol{v},w,y},\set{\ol{v},\ol{y}},\set{\ol{w},x_1},\dots,\set{\ol{w},x_p},\set{\ol{x_1},\dots,\ol{x_p}}}$ for $p \ge 1$. For $(v,w)$ we have $(2,p)$, while for $(w,v)$ for have $(p,p+1)$.
    \end{enumerate}
  \item Finally we consider Part \ref{thm:nbesp23}, showing that $m_{i+1}' = k$ for all $k \in \NN$ is possible.
    \begin{enumerate}
    \item For $k=1$ consider $F_3 := \set{\set{v},\set{\ol{v},w},\set{\ol{v},\ol{w}}}$. For $(v,w)$ we have $(2,1)$, and for $(w,v)$ we have $(1,1)$.
    \item Let $F_4 := \set{\set{v},\set{\ol{v},x_1},\dots,\set{\ol{v},x_k},\set{w,\ol{x_1},\dots,\ol{x_k}},\set{\ol{w},\ol{x_1},\dots,\ol{x_k}}}$ for $k \ge 2$. For $(v,w)$ we have $(k,1)$, and for $(w,v)$ we have $(1,k)$.
    \end{enumerate}
  \end{enumerate}
\end{example}
Finally we give examples showing that the bounds from Part \ref{thm:nbesp3} (the case $m_i = 1$) of Theorem \ref{thm:nbesp} are sharp in general.
\begin{example}\label{exp:nbesp3}
   All examples (again) are in $\Musati{\delta=1}$.
  \begin{enumerate}
  \item For Part \ref{thm:nbesp31} ($m_{i+1} = 1$), that is, the singularity-degree tuple $(1,1)$, it is trivial that after swap we can have $(1,1)$ again, while to obtain $(2,1)$ consider $F_3$ from Example \ref{exp:nbesp2} in the other direction.
  \item Consider Part \ref{thm:nbesp32} ($m_{i+1} \ge 2$).
    \begin{enumerate}
    \item An example showing that the swap can be impossible is given by $F := \set{\set{v,w},\set{\ol{v},w},\set{\ol{w},x_1},\dots,\set{\ol{w},x_k},\set{\ol{x_1},\dots,\ol{x_k}}}$ for $k \ge 2$: For $(v,w)$ we have $(1,k)$, while $(w,v)$ is not singular.
    \item And to obtain swap-results $(1,k) \leadsto (k,k)$ we use Example \ref{exp:turn12}, but in the other direction.
    \end{enumerate}
  \end{enumerate}
\end{example}

\subsubsection{Applications}
\label{sec:exchangelemmaApplications}

We first consider singular tuples where all permutations are also singular:
\begin{definition}\label{def:totallysymmsingt}
  Consider $F \in \Musat$ and a tuple $\vec{v} = (v_1,\dots,v_n)$ ($n \in \NNZ$). $\vec{v}$ is called \textbf{totally singular} for $F$ if $\vec{v}$ is singular for $F$ with $\sgp(F, (v_1,\dots,v_n)) = S_n$.
\end{definition}
\begin{corollary}\label{cor:all2}
  Consider $F \in \Musat$ and a singular tuple $\vec{v} = (v_1,\dots,v_n)$ ($n \in \NNZ$) such that each $v_i$ is non-1-singular in $F$ (i.e., $\set{v_1,\dots,v_n} \sse \varnosing(F)$). Then $\vec{v}$ is totally singular for $F$, and for each permutation $\vec{v}'$ every variable is non-1-singular for $\vec{v}$.
\end{corollary}
\begin{proof}
  With Part \ref{thm:nbesp21} of Theorem \ref{thm:nbesp} and Part \ref{cor:propdps311} of Corollary \ref{cor:propdps3}. \qed
\end{proof}
We remark that in the conference version, that is Corollary 27 in \cite{KullmannZhao2012ConfluenceC}, a more general version is stated, only assuming for $\vec{v}$ that every variable is not-1-singular for it (not, as in Corollary \ref{cor:all2}, already for $F$). We believe this more general statement is true, but the proof there is false. The more general version is not needed for any of the other results of \cite{KullmannZhao2012ConfluenceC} or this report. Furthermore a false additional assertion is given in Corollary 27 in \cite{KullmannZhao2012ConfluenceC}, namely that all permutation of $\vec{v}$ would also be non-1-singular, which is refuted by the following example.
\begin{example}\label{exp:non1singniperm}
  Consider $F := \set{\set{v,a},\set{\ol{a}},\set{\ol{v},b},\set{\ol{v},\ol{b}}} \in \Musati{\delta=1}$. Then $(v,a)$ has the property that all variables are non-1-singular for it, while $(a)$ is 1-singular for $F$.
\end{example}

We mention another (simpler) case of total singularity (which already follows by Corollary \ref{cor:propdps3}, Part \ref{cor:propdps321}):
\begin{corollary}\label{cor:all1}
  Consider $F \in \Musat$ and a singular tuple $\vec{v} = (v_1,\dots,v_n)$ such that $\set{v_1,\dots,v_n} \sse \varosing(F)$. Then $\vec{v}$ is totally singular, and for each permutation $\vec{v}'$ of $\vec{v}$ each variable is 1-singular (for $\vec{v}'$).
\end{corollary}
\begin{proof}
  With Part \ref{thm:nbesp311} of Theorem \ref{thm:nbesp} and Part \ref{cor:propdps321} of Corollary \ref{cor:propdps3}. \qed
\end{proof}
Finally we get some normal form of a singular tuple $\vec{v}$ for $F \in \Musat$ by moving the singular variables from $F$ to the front, followed by further 1-singular DP-reductions, and concluded by non-1-singular DP-reductions:
\begin{corollary}\label{cor:first1}
  Consider $F \in \Musat$ and a singular tuple $\vec{v} = (v_1,\dots,v_n)$. Let $V := \set{v_1,\dots,v_n} \cap \varosing(F)$ and $p := \abs{V}$. Consider any $\pi_0: \tb 1p \ra \tb 1n$ such that $\set{v_{\pi_0(i)} : i \in \tb 1p} = V$. Then there exists $q \in \tb pn$ and an s-preserving permutation $\pi$ for $\vec{v}$ such that $\pi$ extends $\pi_0$, and $v_{\pi(i)}$ is 1-singular for $(v_{\pi(1)},\dots,v_{\pi(n)})$ and $i \in \tb 1n$ if and only if $i \le q$.
\end{corollary}
\begin{proof}
  The sorting of $\vec{v}$ is computed via singularity-preserving neighbour swaps, in four steps (``processes''). Process I establishes that in the associated singularity-degree tuple all entries equal to $1$ appear in the front-part (the first $q$ elements). This is achieved by noting that a neighbouring degree-pair $(\ge 2, 1)$ can be swapped and becomes $(1,\ge 1)$. Thus we can grow the 1-singular front part by every value $1$ occurring not in it, and we obtain a permutation where all singularity-degrees of value $1$ appear in the (consecutive) front-part (while the back-part has all singularity-degrees of values $\ge 2$).

Process II now additionally moves variables in $V$ occurring in the back-part to the front-part as follows: If there is still such a variable, then this can not be the first place in the back-part, and so the variable can be moved one place to the left. Possibly process I has to applied after this step (if it does, then the front-part grows at least by one element). This process can be repeated and terminates once all of $V$ is in the front part. Now the variables in the front part and especially $q$ have been determined. In the remainder the front part is put into a suitable order.

Process III only considers the front part, and the task is to move all variables in $V$ to its front. This is unproblematic, since 1-singular DP-reduction does not increase literal degrees. Finally process IV commutes the variables in $V$ into the given order. \qed
\end{proof}
Comparing two different singular tuples, they don't need to overlap, however they need to have a ``commutable beginning'' via appropriate permutations, given they contain at least two variables:
\begin{lemma}\label{lem:comptwost}
  Consider $F \in \Musat$ and singular tuples $(v_1,\dots,v_p)$, $(w_1,\dots,w_q)$ for $F$ with $p, q \ge 2$. Then there is an s-preserving permutation $\pi$ for $(v_1,\dots,v_p)$ and an s-preserving permutation $\pi'$ for $(w_1,\dots,w_q)$, such that both $(v_{\pi(1)},w_{\pi'(1)})$ and $(w_{\pi'(1)},v_{\pi(1)})$ are singular for $F$.
\end{lemma}
\begin{proof}
If one of the two tuples contains a 1-singular variable $v_i \in \varosing(F)$ resp.\ $w_i \in \varosing(F)$, then the assertion follows by Corollary \ref{cor:first1} and Part \ref{cor:propdps32} of Corollary \ref{cor:propdps3}. So assume that neither contains a 1-singular variable from $F$. Note that if none of the variables of a singular tuple is 1-singular for $F$, then all the variables in it must be singular for $F$, since new singular variables are only created by 1-singular DP-reduction according to Corollary \ref{cor:propdps3}, Part \ref{cor:propdps311}. Thus the assertion follows by Corollary \ref{cor:all2} and Lemma \ref{lem:22sing}. \qed
\end{proof}

\subsection{Without 1-singular variables}
\label{sec:without1sing}

If $F \in \Musat$ has no 1-singular variables, then we know its maximal singular tuples (singular tuples which can not be extended), as we will show in Lemma \ref{lem:Fonly2}, namely they are given by choosing exactly one singular literal from each clause which contains singular literals. In this context the concept of ``singularity hypergraph'' is useful, so that we can recognise such maximal singular tuples as minimal ``transversals''. Recall that a \emph{hypergraph} $G$ is a pair $G = (V,E)$, where $V$ is a set, the elements called ``vertices'', while $E$ is a set of subsets of $V$, the elements called ``hyperedges''; the notations $V(G) := V$ and $E(G) := E$ are used.
\begin{definition}\label{def:shyp}
  For $F \in \Musat$ we define the \textbf{singularity hypergraph} \bmm{\shyp(F)} as follows:
  \begin{itemize}
  \item The vertex set is $\var(F)$ (the variables of $F$).
  \item For every $v \in \varsing(F)$ let $x_v$ be the singular literal (which depends on the given choice in case $v$ is 1-singular), and let $L := \set{x_v : v \in \varsing(F)}$.
  \item Now the hyperedges are given by $\var(C \cap L)$ for $C \in F$ with $C \cap L \not= \es$.
  \end{itemize}
  I.e.,
  \begin{displaymath}
    \shyp(F) := (\var(F), \, \set{\var(C \cap L) : C \in F \und C \cap L \not= \es}).
  \end{displaymath}
  Note that the hyperedges of $\shyp(F)$ are non-empty and pairwise disjoint.
\end{definition}
\begin{example}\label{exp:conflsdp_cont1}
  Continuing Example \ref{exp:conflsdp}:
  \begin{enumerate}
  \item For $F$ as in Part \ref{exp:conflsdp1} we have $\shyp(F) = (\set{v,v_1,v_2}, \set{\set{v,v_1}})$.
  \item For $F$ as in Part \ref{exp:conflsdp2} we have $\shyp(F) = (\set{v,w,v_1,v_2,v_1',v_2'}, \set{\set{v,v_1}})$.
  \end{enumerate}
\end{example}
\begin{example}\label{exp:shyp2}
  With another inverse sDP-reduction, applied to $F$ from Part \ref{exp:conflsdp1} of Example \ref{exp:conflsdp} and introducing variable $v'$, we obtain
  \begin{displaymath}
    F = \set{ \set{v,v_1}, \set{\ol{v},v_2}, \set{\ol{v},\ol{v_2}}, \set{v',\ol{v_1}}, \set{\ol{v'},v_2}, \set{\ol{v'},\ol{v_2}} }.
  \end{displaymath}
  We have $\varsing(F) = \set{v_1,v,v'}$ and $\varosing(F) = \set{v_1}$. Choosing $v_1$ resp.\ $\ol{v_1}$ as the singular literal for $v_1$, we have $\shyp(F) = (\set{v,v',v_1,v_2}, \set{\set{v,v_1},\set{v'}})$ resp.\ $= (\set{v,v',v_1,v_2}, \set{\set{v},\set{v',v_1}})$.
\end{example}
\begin{example}\label{exp:shyp3}
  Consider
  \begin{displaymath}
    F := \setb{ \set{a,b}, \set{\ol{a},x,v},\set{\ol{a},y,v'}, \set{\ol{b},x,v}, \set{\ol{b},y,v'}, \set{\ol{x},v},\set{\ol{y},v'}, \set{\ol{v},\ol{v'}} }.
  \end{displaymath}
  We have $\shyp(F) = (\set{a,b,x,y,v,v'}, \set{\set{a,b},\set{x},\set{y},\set{v,v'}})$. We have furthermore the properties $F \in \Musati{\delta=2} \sm \Smusati{\delta=2}$ and $\var(F) = \varnosing(F)$.
\end{example}

\begin{definition}\label{def:maxsingtup}
  Consider $F \in \Musat$. A singular tuple $(v_1,\dots,v_n)$ for $F$ is called \textbf{maximal}, if there is no singular tuple extending it, i.e., $\dpi{v_1,\dots,v_n}(F)$ is nonsingular.
\end{definition}

\begin{lemma}\label{lem:Fonly2}
  Consider $F \in \Musat$ with $\varosing(F) = \es$. The variable-sets of maximal singular tuples for $F$ are precisely the minimal transversals of $\shyp(F)$ (minimal sets of vertices intersecting every hyperedge). And the maximal singular tuples of $F$ are precisely obtained as (arbitrary) linear orderings of these variable-sets.
\end{lemma}
\begin{proof}
By Corollary \ref{cor:propdps3}, Part \ref{cor:propdps311}, for each singular tuple $(v_1,\dots,v_n)$ of $F$ we have $\set{v_1,\dots,v_n} \sse \varnosing(F) = \varsing(F)$. So by Corollary \ref{cor:all2} all permutations are singular. Finally, for $v \in \varsing(F)$ let $F_v := \dpi{v}(F)$, let $C_v \in F$ be the main clause of $v$, and let $H_v := \var(C_v) \cap \varsing(F)$. Then we have $\shyp(F_v) = (V(\shyp(F)) \sm \set{v}, E(\shyp(F)) \sm \set{H_v})$. The assertion of the lemma follows now easily by induction. \qed
\end{proof}
\begin{example}\label{exp:conflsdp_cont2}
  Continuing Example \ref{exp:conflsdp} (and Example \ref{exp:conflsdp_cont1}): For $F$ as in Part \ref{exp:conflsdp1} as well as in Part \ref{exp:conflsdp2} the two maximal singular tuples are $(v)$ and $(v_1)$.
\end{example}
\begin{example}\label{exp:shyp3_cont}
  Continuing Example \ref{exp:shyp3}: We have $2 \cdot 2 = 4$ minimal transversals, namely $\set{a,x,y,v}, \set{b,x,y,v}, \set{a,x,y,v'},\set{b,x,y,v'}$. There are thus $4$ elements in $\sdp(F)$; Theorem \ref{thm:confmodisomu2} will show that they are necessarily all isomorphic to $\Dt{2}$ (since after reduction $2$ variables remain; recall Example \ref{exp:F2F3}). Finally we remark that $F$ has precisely $4 \cdot 4! = 96$ maximal singular tuples.
\end{example}
Since two different minimal transversals of $\shyp(F)$ remove different variables, they result in different sDP-reduction results. So the elements of $\sdp(F)$ are here in bijective correspondence to the minimal transversals of $F$, and we get:
\begin{corollary}\label{cor:numsDP}
  For $F \in \Musat$ with $\varosing(F) = \es$ we have that $\abs{\sdp(F)}$ is the number of minimal transversals of $\shyp(F)$.
\end{corollary}

\subsection{The singularity index}
\label{sec:singindex}

\begin{definition}\label{def:singindex}
  Consider $F \in \Musat$. The \textbf{singularity index} of $F$, denoted by $\bmm{\singind(F)} \in \NNZ$, is the minimal $n \in \NNZ$ such that a maximal singular tuple of length $n$ exists for $F$.
\end{definition}
So $\singind(F) = 0 \Lra F \in \Musatns$. See Corollary \ref{cor:singindevsat}, Part \ref{cor:singindevsat1}, for a characterisation of $F \in \Musat$ with $\singind(F) = 1$. In Theorem \ref{thm:singind} we see that all maximal singular tuples are of the same length (given by the singularity index). By Lemma \ref{lem:Fonly2} we get:
\begin{lemma}\label{lem:singindonly2}
  Consider $F \in \Musat$ not having 1-singular variables (i.e., $\varosing(F) = \es$). Then every maximal singular tuple has length $\singind(F)$, which is the number of different clauses of $F$ containing at least one singular literal.
\end{lemma}

More general than Lemma \ref{lem:singindonly2} (but with less details), we show next that for all minimally unsatisfiable clause-sets all maximal singular tuples (i.e., maximal sDP-reduction sequences) have the same length. The basic idea is to utilise the good commutativity properties of 1-singular variables, so that induction on the singularity index can be used.
\begin{theorem}\label{thm:singind}
  For $F \in \Musat$ and every maximal singular tuple $(v_1,\dots,v_m)$ for $F$ we have $m = \singind(F)$.
\end{theorem}
\begin{proof}
We prove the assertion by induction on $\singind(F)$. For $\singind(F) = 0$ the assertion is trivial, so assume $\singind(F) > 0$. If $F$ has no 1-singular variables, then the assertion follows by Lemma \ref{lem:singindonly2}, and so we assume that $F$ has a 1-singular variable $v$. First we show that we can choose $v$ such that $\singind(\dpi{v}(F)) = n-1$.

 Consider a maximal singular tuple $(v_1,\dots,v_n)$ of length $n = \singind(F)$. Note that $\singind(\dpi{v_1}(F)) = n-1$. If $v_1$ is 1-singular, then we can use $v := v_1$ and we are done, and so assume $v_1$ is not 1-singular. The induction hypothesis, applied to $\dpi{v_1}(F)$, yields $\singind(\dpi{v_1,v}(F)) = n-2$. Now by Corollary \ref{cor:propdps3}, Part \ref{cor:propdps32}, both tuples $(v_1,v)$ and $(v,v_1)$ are singular for $F$, whence $\dpi{v_1,v}(F) = \dpi{v,v_1}(F)$ holds (Corollary \ref{cor:singeqvar}), and so $\singind(\dpi{v,v_1}(F)) = n-2$. We obtain $\singind(\dpi{v}(F)) \le n-1$, and thus $\singind(\dpi{v}(F)) = n-1$ as claimed.

Now consider an arbitrary maximal singular tuple $(w_1,\dots,w_m)$. It suffices to show that $\singind(\dpi{w_1}(F)) \le n-1$, from which by induction hypothesis the assertion follows. The argument is now similar to above. The claim holds for $w_1=v$, and so assume $w_1 \not= v$. By induction hypothesis we have $\singind(\dpi{v,w_1}(F)) = n-2$. By Corollary \ref{cor:propdps3}, Part \ref{cor:propdps32}, both tuples $(v,w_1)$ and $(w_1,v)$ are singular for $F$. Thus $\singind(\dpi{w_1,v}(F)) = n-2$. We obtain $\singind(\dpi{w_1}(F)) \le n-1$ as claimed. \qed
\end{proof}

\begin{corollary}\label{cor:singindsamen}
  For $F \in \Musat$ and $F', F'' \in \sdp(F)$ we have $n(F') = n(F'')$.
\end{corollary}

\section{Confluence modulo isomorphism on eventually $\Smusat$}
\label{sec:confmodiso}

Finally we are able to show our third major result, confluence modulo isomorphism of singular DP-reduction in case all maximal sDP-reductions yield saturated clause-sets.
\begin{definition}\label{def:eventsaturated}
  A minimally unsatisfiable clause-set $F$ is called \textbf{eventually saturated}, if all nonsingular $F'$ with $F \tsdps F'$ are saturated; the set of all eventually saturated clause-sets is $\bmm{\Esmusat} := \set{F \in \Musat : \sdp(F) \sse \Smusat}$.
\end{definition}
By Corollary \ref{cor:sdppressmu} we have $\Smusat \sse \Esmusat$. If $\mc{C} \sse \Musat$ is stable under sDP-reduction, then we have $\mc{C} \sse \Esmusat$ iff $\mc{C} \cap \Musatns \sse \Smusat$. In order to show $\Esmusat \sse \cflimusat$ (recall Definition \ref{def:conflsdp}), we show first that ``divergence in one step'' is enough, that is, if we have a clause-set $F \in \Musat$ such that sDP-reduction is not confluent modulo isomorphism, then we can obtain from $F$ by sDP-reduction the clause-set $F' \in \Musat$ with singularity index $1$ (thus using $\singind(F)-1$ reduction steps) such that also for $F'$ sDP-reduction is not confluent modulo isomorphism:
\begin{lemma}\label{lem:dev1s}
  Consider $F \in \Musat \sm \cflimusat$. So $\singind(F) \ge 1$. Then there is a singular tuple $(v_1,\dots,v_{\singind(F)-1})$ for $F$, such that for $F' := \dpi{v_1,\dots,v_{\singind(F)-1}}(F)$ we still have $\sdp(F') \in \Musat \sm \cflimusat$ (note $\singind(F') = 1$).
\end{lemma}
\begin{proof}
We prove the assertion by induction on $\singind(F) \ge 1$. The assertion is trivial for $\singind(F) = 1$, and so consider $n := \singind(F) \ge 2$. If there is a singular variable $v \in \varsing(F)$ with $\dpi{v}(F) \in \Musat \sm \cflimusat$, then the assertion follows by induction hypothesis. So assume for the sake of contradiction, that for all singular variables $v$ we have $\dpi{v}(F) \in \cflimusat$. Consider (maximal) singular tuples $(v_1,\dots,v_n), (w_1,\dots,w_n)$ for $F$ such that $\dpi{\vec{v}}(F)$ and $\dpi{\vec{w}}(F)$ are not isomorphic. By Lemma \ref{lem:comptwost} w.l.o.g.\ we can assume that $(v_1,w_1)$ and $(w_1,v_1)$ are both singular for $F$, whence $\dpi{v_1,w_1}(F) = \dpi{w_1,v_1}(F)$ by Corollary \ref{cor:singeqvar}. We have $\dpi{v_1}(F), \dpi{w_1}(F) \in \cflimusat$ by assumption, and we obtain the contradiction that $\dpi{\vec{v}}(F)$ and $\dpi{\vec{w}}(F)$ are isomorphic, since $\dpi{\vec{v}}(F)$ is isomorphic to the result obtained by reducing $F$ via a (maximal) singular tuple $\vec{v}' = (v_1,w_1,\dots)$ of length $n$, where permuting the first two elements in $\vec{v'}$ yields the singular tuple $\vec{w}' = (w_1,v_1,\dots)$ with the same result, which in turn is isomorphic to $\dpi{\vec{w}}(F)$. \qed
\end{proof}

\begin{corollary}\label{cor:div1s}
  Consider a class $\mc{C} \sse \Musat$ which is stable under application of singular DP-reduction. Then we have $\mc{C} \sse \cflimusat$ if and only if $\set{F \in \mc{C} : \singind(F)=1} \sse \cflimusat$.
\end{corollary}
Now we analyse the main case where all sDP-reductions give saturated results:
\begin{lemma}\label{lem:isonssmu}
  Consider $F \in \Musat$ and a clause $C \in F$. Let $C' := \set{x \in C : \ldeg_F(x) = 1}$ be the set of singular literals in $C$, establishing $C$ as the main clause for the underlying singular variables $\var(x)$ (for $x \in C'$), and let $F_x := \set{D \in F : \ol{x} \in D}$ be the set of side clauses of $\var(x)$ for $x \in C'$. Due to $F \in \Musat$ the sets $F_x$ are non-empty and pairwise disjoint (note that $\var(x)$ is $\abs{F_x}$-singular in $F$ for $x \in C'$). Now assume $\abs{C'} \ge 2$, and that for all $x \in C'$ we have $\dpi{\var(x)}(F) \in \Smusat$. Then:
  \begin{enumerate}
  \item\label{lem:isonssmu1} $\abs{C'} = 2$.
  \item\label{lem:isonssmu2} $\fa\, x \in C'\, \fa\, D \in F_x : (C \sm C') \sse D$.
  \item\label{lem:isonssmu3} For $x, y \in C'$ we have that $\dpi{\var(x)}(F)$ and $\dpi{\var(y)}(F)$ are isomorphic.
  \end{enumerate}
\end{lemma}
\begin{proof}
Consider (any) literals $x, y \in C'$ with $x \not= y$. Then for $D \in F_x$ we have $(C \sm \set{x,y}) \sse D$ by Corollary \ref{cor:sdppartsatur}, since otherwise the corollary can be applied to $\var(x)$, replacing $D$ by $D \cup  (C \sm \set{x,y})$, which yields the partial saturation $F' \in \Musat$ of $F$ with singular variable $\var(y)$, and where then $\dpi{\var(y)}(F')$ would yield a proper partial saturation $G$ of $\dpi{\var(y)}(F)$, contradicting that the latter is saturated. It follows that actually $C' = \set{x,y}$ must be the case, since if there would be $z \in C' \sm \set{x,y}$, then $\ldeg_F(z) \ge 2$ contradicting the definition of $C'$. It follows Part \ref{lem:isonssmu2}. Finally for Part \ref{lem:isonssmu3} we note that now $F \leadsto \dpi{x}(F)$ just replaces $\ol{x}$ in the clauses of $F_x$ by $y$, while $F \leadsto \dpi{y}(F)$ just replaces $\ol{y}$ in the clauses of $F_y$ by $x$, and thus renaming $y$ in $\dpi{x}(F)$ to $\ol{x}$ yields $\dpi{y}(F)$. \qed
\end{proof}

\begin{corollary}\label{cor:singindevsat}
  For $F \in \Musat$ with $\singind(F) = 1$ we have:
  \begin{enumerate}
  \item\label{cor:singindevsat1} If $\abs{\varsing(F)} \ge 2$:
    \begin{enumerate}
    \item\label{cor:singindevsat11} $\varsing(F) = \varnosing(F)$, that is, all singular variables are non-1-singular.
    \item\label{cor:singindevsat12} The main clauses of the singular variables coincide (that is, there is $C \in F$ such that for all singular literals $x$ for $F$ we have $x \in C$).
    \item\label{cor:singindevsat13} If $F \in \Esmusat$ then $\abs{\varsing(F)} = 2$.
    \end{enumerate}
  \item\label{cor:singindevsat2} If $F \in \Esmusat$ then $F \in \cflimusat$.
  \end{enumerate}
\end{corollary}
\begin{proof}
Part \ref{cor:singindevsat11} follows by Part \ref{cor:propdps321} of Corollary \ref{cor:propdps3}, and Part \ref{cor:singindevsat12} follows by Lemma \ref{lem:22sing}. Now Parts \ref{cor:singindevsat13}, \ref{cor:singindevsat2} follow from Lemma \ref{lem:isonssmu}. \qed
\end{proof}
\begin{example}\label{exp:conflsdp_cont3}
  The two clause-sets $F$ from Example \ref{exp:conflsdp} (recall Example \ref{exp:conflsdp_cont2}) fulfil $\singind(F) = 1$ and $\abs{\varsing(F)} = 2$. For $F$ from in Part \ref{exp:conflsdp1} there we have $F \in \Esmusat$, for $F$ from Part \ref{exp:conflsdp2} we have $F \notin \cflimusat$.
\end{example}

By Corollary \ref{cor:div1s} we obtain from Part \ref{cor:singindevsat2} of Corollary \ref{cor:singindevsat}:
\begin{theorem}\label{thm:evsatcfl}
  $\Esmusat \subset \cflimusat$.
\end{theorem}

\section{Applications to $\Musati{\delta=2}$}
\label{sec:appmu2}

If $F \in \cflimusat$, then we can speak of \emph{the non-singularity type} of $F$ as the (unique) isomorphism type of the elements of $\sdp(F)$. In this section we show that for $F \in \Musati{\delta=2}$ these assumptions are fulfilled.  First we recall the fundamental classification:
\begin{definition}\label{def:mud2}
  Consider $n \ge 2$, let addition for the indices of variables $v_1, \dots, v_n$ be understood modulo $n$ (so $n+1 \leadsto 1$), and define $P_n := \set{v_1, \dots, v_n}$, $N_n := \set{\ol{v_1}, \dots, \ol{v_n}}$, $C_i := \set{\ol{v_i},v_{i+1}}$ for $i \in \tb 1n$, and finally $\bmm{\Dt{n}} := \setb {P_n, N_n} \cup \setb{C_i : i \in \tb 1n} \in \Musatnsi{\delta=2}$.
\end{definition}
So $n(\Dt{n}) = n$ and $c(\Dt{n}) = n+2$. Recall Example \ref{exp:F2F3}, where $\Dt{2}, \Dt{3}, \Dt{4}$ were already given. The clause-sets $\Dt{n}$ are precisely (up to isomorphism) the non-singular elements of $\Musati{\delta=2}$:
\begin{theorem}\label{thm:HKBMU2}\cite{KleineBuening2000SubclassesMU}
  For $F \in \Musatnsi{\delta=2}$ we have $F \cong \Dt{n(F)}$.
\end{theorem}
We show now that for $F \in \Musati{\delta=2}$ we have the non-singularity type of $F$, which can be encoded as the number of variables left after complete sDP-reduction, using that the isomorphism types in $\Musatnsi{\delta=2}$ are determined by their number of variables:
\begin{theorem}\label{thm:confmodisomu2}
  $\Musati{\delta=2} \sse \cflimusat$.
\end{theorem}
\begin{proof}
  The first proof is obtained by applying Corollary \ref{cor:singindsamen} and the observation that non-isomorphic elements of $\Musatnsi{\delta=2}$ have different numbers of variables. The second proof is obtained by applying Theorem \ref{thm:evsatcfl} and the fact that $\Musatnsi{\delta=2} \sse \Smusat$, whence $\Musati{\delta=2} \sse \Esmusat$. \qed
\end{proof}

\begin{definition}\label{def:typeMU2}
  By Theorem \ref{thm:confmodisomu2} to every $F \in \Musati{\delta=2}$ we can associate its \textbf{non-singularity type} $\bmm{\mutt(F)} \in \NN_{\ge 2}$, the unique $n$ such that $F$ by singular DP-reduction can be reduced to a clause-set isomorphic to $\Dt{n}$.
\end{definition}
So, considering the structure of $\Dt{n}$ as a ``contradictory cycle'', we can say that every $F \in \Musati{\delta=2}$ contains a contradictory cycle, where the length of that cycle is $\mutt(F)$ (and thus uniquely determined), while, as Example \ref{exp:conflsdp} shows, the variables constituting such a cycle are not uniquely determined.

\section{Conclusion and open problems}
\label{sec:open}

We have discussed questions regarding confluence of singular DP-reduction on minimally unsatisfiable clause-sets. Besides various detailed characterisations, we obtained the invariance of the length of maximal sDP-reduction-sequences, confluence for saturated and confluence modulo isomorphism for eventually saturated clause-sets. The main open questions regarding these aspects are:
\begin{enumerate}
\item Can we obtain a better overview on singular tuples for $F \in \Musat$ ?
  \begin{enumerate}
  \item What are the structural properties of the set of all singular tuples, for $F \in \Musat, \Smusat, \Uclash$ ?
  \item Especially for $F \in \Uclash$ it should hold that if $\singind(F)$ is ``large'', then $\abs{\varsing(F)}$ must be ``large''. More precisely:
    \begin{conjecture}
      For every $k \in \NN$ there are $a \in \NN$ and $\alpha \in \RR_{>0}$ such that for all $F \in \Uclashi{\delta=k}$ with $\singind(F) \ge a$ we have $\abs{\varsing(F)} \ge \alpha \cdot \singind(F)$.
    \end{conjecture}
  \end{enumerate}
\item Can we characterise $\cflmusat$ and/or $\cflimusat$? Especially, what is the decision complexity of these classes?
\item Are there other interesting classes for which we can show confluence resp.\ confluence mod isomorphism of singular DP-reduction?
\end{enumerate}
As a first application of our results, in Subsection \ref{sec:appmu2} we considered the types of (arbitrary) elements of $\Musati{\delta=2}$. This detailed knowledge is a stepping stone for the determination of the isomorphism types of the elements of $\Musatnsi{\delta=3}$, which we have obtained meanwhile (to be published; based on a mixture of general insights into the structure of $\Musat$ and detailed investigations into $\Musati{\delta\le2}$).

The major open problem of the field is the classification (of isomorphism types) of $\Musatnsi{\delta=k}$ for arbitrary $k$. The point of departure is the conjecture stated in \cite{KullmannZhao2011Bounds} that for $F \in \Uclashnsi{\delta=k}$ the number $n(F)$ of variables is bounded.

Regarding the potential applications from Subsection \ref{sec:introapp}, applying singular DP-reductions in algorithms searching for MUS's is a natural next step.

Finally, a promising direction is the generalisation of the results of this paper beyond minimal unsatisfiability, possibly to arbitrary clause-sets: The analysis of sDP-reduction is much simplified by the fact that for $F \in \Musat$ all possible resolutions must actually occur, without producing tautologies and without producing any contractions. To handle arbitrary $F \in \Cls$, these complications have to be taken into account.

\bibliographystyle{plain}

\newcommand{\noopsort}[1]{}

\end{document}